\documentclass[11pt]{article}
\usepackage[margin=1in]{geometry}

\usepackage{amsmath,amssymb,amsthm,enumerate}
\usepackage{comment}
\usepackage{graphicx}
\usepackage{hyperref}

\newtheorem{theorem}{Theorem}
\newtheorem{lemma}[theorem]{Lemma}
\newtheorem{corollary}[theorem]{Corollary}
\newtheorem{claim}[theorem]{Claim}

\theoremstyle{definition}

\renewcommand{\AA}{\mathcal{A}}

\newcommand{\CC}{\mathcal{C}}

\newcommand{\GG}{\mathcal{G}}
\newcommand{\HH}{\mathcal{H}}
\newcommand{\MM}{\mathcal{M}}
\newcommand{\OO}{\mathcal{O}}

\newcommand{\UU}{\mathcal{U}}
\newcommand{\ZZ}{\mathbb{Z}}
\renewcommand{\SS}{\mathcal{S}}
\newcommand{\Exp}{\ensuremath{\operatorname{\mathbf{E}}}}

\makeatletter
\newcommand*{\bdiv}{%
  \nonscript\mskip-\medmuskip\mkern5mu%
  \mathbin{\operator@font div}\penalty900\mkern5mu%
  \nonscript\mskip-\medmuskip
}
\makeatother

\usepackage[textsize=tiny,disable]{todonotes}

\newcommand{\gtodolight}[1]{\todo[color=red!25]{#1}}

\title{Amplifiers and Suppressors of Selection for  \\ the Moran Process on Undirected Graphs}
\author{George Giakkoupis\\
INRIA, France\\
george.giakkoupis@inria.fr
}

\begin{document}

\maketitle

\begin{abstract}
We consider the classic Moran process modeling the spread of genetic mutations, as extended to structured populations by Lieberman et al.\ (Nature, 2005).
In this process, individuals are the vertices of a connected graph $G$.
Initially, there is a single mutant vertex, chosen uniformly at random.
In each step, a random vertex is selected for reproduction with a probability proportional to its fitness:
mutants have fitness $r>1$, while non-mutants have fitness~$1$.
The vertex chosen to reproduce places a copy of itself to a uniformly random neighbor in $G$, replacing the individual that was there.
The process ends when the mutation either reaches fixation (i.e., all vertices are mutants), or gets extinct.
The principal quantity of interest is the probability with which each of the two outcomes occurs.

A problem that has received significant attention recently concerns the existence of families of graphs, called strong amplifiers of selection, for which the fixation probability tends to~1 as the order $n$ of the graph increases, and the existence of strong suppressors of selection, for which this probability tends to 0.
For the case of directed graphs, it is known that both strong amplifiers and suppressors exist.
For the case of undirected graphs, however, the problem has remained open, and the general belief has been that neither strong amplifiers nor suppressors exist.
In this paper we disprove this belief, by providing the first examples of such graphs.
The strong amplifier we present has fixation probability $1-\tilde O(n^{-1/3})$, and the strong suppressor has fixation probability $\tilde O(n^{-1/4})$.
Both graph constructions are surprisingly simple.
We also prove a general upper bound of $1-\tilde \Omega(n^{-1/3})$ on the fixation probability of any undirected graph.
Hence, our strong amplifier is existentially optimal.

%\paragraph{Keywords:} Evolutionary graph theory
\end{abstract}

%\vspace{20ex}
%\paragraph{Note to the Reviewers:} Sections marked with an asterisk (*) can be skipped without loss of continuity; the remaining sections are within the 10-page limit.
%
%\setcounter{page}{0}
%\thispagestyle{empty}
%\clearpage

\section{Introduction}

Evolutionary dynamics have been traditionally studied in the context of well-mixed populations.
One of the simplest and most influential stochastic models of evolution is the \emph{Moran process}, introduced in the 1950s~\cite{moran1958random,moran1962statistical}.
It is based on a finite population of fixed size $n$, where each individual is classified either as a \emph{mutant} or a \emph{non-mutant}.
Each type is associated with a \emph{fitness} parameter, which determines its rate of reproduction.
Mutants have relative fitness $r>0$, as compared to non-mutants, whose fitness is~1
(the mutation is \emph{advantageous} if $r>1$).
In each step of the process, a random individual is chosen for reproduction with a probability proportional to its fitness, and produces an offspring which is a copy of itself.
This offspring replaces an individual selected uniformly at random among the population.
The process ends when the mutation reaches \emph{fixation} (all individuals are mutants), or \emph{extinction}.

The Moran process models the tension between the evolutionary mechanisms of \emph{natural selection} and \emph{genetic drift}.
Selection is captured by the model's property that mutations with larger fitness are more likely to reach fixation; while genetic drift is captured in the probabilistic nature of the model, which allows for advantageous mutants get extinct and for less advantageous ones get fixed.
The main quantity of interest in this model is the probability by which the mutation reaches fixation.
When there are $i$ mutants initially, this probability is
%\begin{equation}
%    \label{eq:fix-homo}
\[    \frac{1-r^{-i}}{1-r^{-n}}, \text{ if } r\neq1, \text{\quad and \quad} \frac in, \text{ if } r=1.
\]
%\end{equation}
%$\frac{1-r^{-i}}{1-r^{-n}}$ if $r\neq1$, and $\frac in$ if $r=1$.

This process does not take into account any spatial or other structure that the evolutionary system may have.
To address this issue,
%Lieberman, Hauert, and Nowak~\cite{lieberman2005evolutionary}
Lieberman et a.~\cite{lieberman2005evolutionary}
proposed a refinement of the model, where individuals are the vertices of an $n$-vertex graph $G$; the graph can be undirected or directed.
As in the original process, the random vertex that reproduces in a step is selected with probability proportional to its fitness.
However, the individual that is replaced by the offspring is selected uniformly at random among the neighbors in $G$ of the reproducing vertex, rather than among the whole population (if $G$ is directed, the outgoing neighbors are considered).
The original Moran process corresponds to the special case where $G$ is the complete graph~$K_n$.

If $G$ is (strongly) connected then the process will terminate with probability 1, reaching either fixation or extinction, as in the original Moran process.
%If $G$ is (strongly) connected then the only possible eventual outcomes are fixation and extinction, as in the original Moran process.
%It was observed in~\cite{lieberman2005evolutionary} that
For the class of regular graph, the fixation probability is identical to that in the original process, depending only on $r$ and the number $i$ of initial mutants.
(This is in fact true for a more general class of weighted graphs~\cite{lieberman2005evolutionary}; see also earlier works~\cite{maruyama1974simple,slatkin1981fixation}.)
However, in general graphs, the probability of fixation depends also on the structure of $G$ and the exact \emph{set} of initial mutants.
A standard aggregate measure, denoted $\rho_G(r)$, is the fixation probability for a single initial mutant, placed at a vertex chosen uniformly at random.
For regular graphs, $\rho_G(r)$ is $\frac{1-r^{-1}}{1-r^{-n}}$ if $r\neq 1$, and $\frac1n$ if $r=1$.
As $n$ grows (keeping $r$ fixed), this probability converges to $1-\frac1r$ if $r>1$, and to $r^{n-1}$ if $r<1$.

For any graph $G$, we have $\rho_G(1) = \frac1n$, by symmetry, as the initial mutant is placed on a random vertex; and for $r<1$, we have $\rho_G(r)\leq \rho_G(1)$, by monotonicity~\cite{diaz2016absorption}.
In the following discussion, we focus mainly on the case of $r>1$, in which the mutation is advantageous.

A graph $G$ is an \emph{amplifier of selection} if $\rho_G(r) > \rho_{K_n}(r)$, for
$r>1$.
In such a graph, it is more likely for an advantageous mutant to reach fixation than in the well-mixed setting, thus $G$ enhances natural selection and
suppresses random drift.
On the other hand, if $\rho_G(r) < \rho_{K_n}(r)$ then $G$ is a \emph{suppressor of selection}, that enhances random drift.
The $n$-vertex star (i.e., the complete bipartite graph $K_{1,n-1}$) is an example of an amplifier, with $\rho_{K_{1,n-1}}(r)$ converging to $1-\frac{1}{r^2}$ as $n$ grows (and $r>1$ stays fixed)~\cite{lieberman2005evolutionary,broom2008analysis}.
An example of a suppressor is more tricky to find.
%Mertzios, Nikoletseas, Raptopoulos, and Spirakis~\cite{mertzios2013natural}
Mertzios et al.~\cite{mertzios2013natural}
described an $n$-vertex graph $G_n$ for which $\lim_{n\to\infty}\rho_{G_n}(r) \leq \frac12\big(1-\frac1r\big)$, for $1<r<\frac43$.

If we do not restrict ourselves to connected undirected graphs or strongly-connected digraphs, but consider also weakly-connected digraphs, then finding a suppressor is easy~\cite{lieberman2005evolutionary}:
For any weakly-connected $G$ with a single source, $\rho_{G}(r) = \frac1n$ for any $r>0$; if $G$ has more sources then $\rho_{G}(r)=0$.

A problem that has received at lot of attention ever since the model was introduced~\cite{lieberman2005evolutionary},
is to identify graph structures for which the probability $\rho_{G}(r)$ attains extreme values.
Let $\GG = \{G_n\}_{n\in I}$, where $I\subseteq \ZZ^+$, be an infinite family of graphs indexed by their number of vertices.
We say that $\GG$ is a \emph{strong amplifier of selection} if $\lim_{n\to\infty} \rho_{G_n}(r) = 1$, for any $r>1$ that is not a function of $n$; while $\GG$ is a \emph{strong suppressor of selection} if $\lim_{n\to\infty} \rho_{G_n}(r) = 0$ instead.
That is, a strong amplifier guarantees that an advantageous mutant will almost surely reach fixation, even if its fitness advantage is small;
whereas in a strong suppressor, the probability for a mutant to be fixed goes to 0,
%no matter how large its fitness advantage is.
however large its fitness advantage. \gtodolight{selective advantage? amplifier: The terminology comes from the fact that the selective advantage of the mutant is being ``amplified" in such graphs.}

Lieberman et al.~\cite{lieberman2005evolutionary} proposed two directed graph families of strong amplifiers, call `superstar' and `metafunnel.'
Formal bounds for the fixation probabilities of these graphs were only recently proved, by
%Galanis, G\"obel, Goldberg, Lapinskas, and Richerby~\cite{galanis2016amplifiers}
Galanis et al.~\cite{galanis2016amplifiers}
(see also~\cite{diaz2013fixation,jamieson2015fixation}).
Precisely, in~\cite{galanis2016amplifiers} a variant of the superstar, called `megastar,' was proposed, and for this family $\{G_n\}$ of graphs it was shown that
$\rho_{G_n} (r) = 1 - \tilde \Theta(n^{-1/2})$, for any $r> 1$ not a function of $n$
(the $\tilde \Theta$ notation hide a factor polynomial in $\log n$).
It was also shown that the megastar is a stronger amplifier than both the superstar and metafunnel.

All the above strong amplifiers are directed graphs, and it appears that they rely on the directionality of their edges in a critical way to achieve high fixation probability.
So far, no undirected families of strong amplifiers have been known.
In fact, no undirected graph family has been known to have fixation probability asymptotically better than the $1-\frac1{r^2}$ probability of the $n$-vertex star.
A related problem for which progress \emph{has} been achieved is to find undirected graph families such that for a significant fraction of the vertices, if the mutation start from any of those vertices it will reach fixation almost surely.
Mertzios and Spirakis~\cite{mertzios2013strong}
identified such a graph family, with the property that for a \emph{constant} fraction of vertices, the probability of fixation is $1-O(1/n)$, for any $r>5$.
Note that if this property were true for a $1-o(1)$ fraction of the vertices (and for any $r>1$), it would imply that the family is a strong amplifier.

Little progress has been made on strong suppressors.
No strong suppressor families of undirected or strongly-connected directed graphs are known.
The only relevant result is that for any $\omega(n) \leq f \leq n^{1/2}$, there is an undirected graph family such that for a $1/f$ fraction of vertices, the probability of fixation when starting from one of those vertices is $O(f/n)$, for any $r$ not a function of $n$~\cite{mertzios2013strong}.

Regarding general bounds on the fixation probability, it is known that for any strongly-connected digraph $G$ with $n$ vertices, $\frac{1}{n} \leq \rho_G(r)\leq 1-\frac{1}{r+n}$, for $r>1$~\cite{diaz2014approximating}.
%(The upper bounds is obtained by bounding the probability that a second mutant is created.)
For the case of undirected graphs, the upper bound has been improved to
$\rho_G(r) = 1-\Omega\big(n^{-2/3}\big)$,  when $r$ is not a function of $n$~\cite{mertzios2013strong}.

Reducing the gap between these upper and lower bounds and the fixation probability of the best known amplifiers and successors, respectively, has been an intriguing problem.
In particular, the lack of progress in the search for undirected strong amplifiers and (undirected or strongly-connected) strong suppressors, has lead to the belief that they do not exist.
Mertzios et al.~\cite{mertzios2013natural,spirakis2013talk} asked whether there exist functions $f_1 (r) > 0$ and $f_2 (r) < 1$ (independent of $n$), such that for any undirected $n$-vertex graph $G$, $f_1(r) \leq \rho_G(r) \leq f_2(r)$.
And at least for the upper bound part, it has been independently suggested that it is true~\cite{goldberg2014talk,diaz2014talk}.

\subsection{Our Contribution}

In this paper, we finally resolve the above questions.
Contrary to the general belief, we show that both undirected strong amplifiers and suppressors exist, and provide the first examples of such graphs.
These examples are strikingly simple, and to some extent ``natural.''
Moreover, we show that our strong amplifier is optimal: We prove a general upper bound on the fixation probability of any undirected graph, which is matched (modulo low-order factors) by the fixation probability  of our amplifier.

Our strong amplifier is a two-parameter family of graphs $A_{n,\epsilon}$.
Parameter $n$ determines the size of the graph, while $1+\epsilon$ is the minimum fitness $r$ for which we guarantee almost sure fixation.
Parameter $\epsilon > 0$ can be a function of $n$ that converges to $0$ as $n$ grows; for convenience we also assume that $\epsilon \leq 1$.

\begin{theorem}
    \label{thm:amplifiers-intro}
    There is a family of undirected graphs $\{A_{n,\epsilon}\}_{n\geq1, 0 < \epsilon \leq 1}$, such that graph $A_{n,\epsilon}$ has $\Theta(n)$ vertices, and  $\rho_{A_{n,\epsilon}}(r) = 1-O\left(\frac{\log n}{\epsilon n^{1/3}}\right)$ for $r \geq 1+\epsilon$.
\end{theorem}

By letting $\epsilon = 1/g(n)$, for any slowly increasing function $g(n)$, e.g., $g(n) = \log\log n$, we obtain a strong amplifier that achieves a fixation probability of $1-O(g(n)\, n^{-1/3}\log n)$ for any $r>1$ that is not a function of $n$.
On the other hand, letting $\epsilon = g(n)\,n^{-1/3}\log n$ yields a strong amplifier achieving a fixation probability of $1-O(1/g(n)) = 1-o(1)$ for $r \geq 1 + g(n)\,n^{-1/3}\log n$; i.e., we have almost sure fixation even when the mutation has an \emph{extremely small} fitness advantage, $r-1$, that is polynomial in $1/n$.

%As already mentioned, no undirected strong amplifiers have been known until now.
The best known \emph{directed} strong amplifier is the megastar, which has a fixation probability of $1-O(n^{-1/2}\log^{23}n)$, for $r>1$ not a function of $n$~\cite{galanis2015full}. %,galanis2016amplifiers}.
This is better than the $1-O(n^{-1/3}\log n)$ fixation probability of our amplifier. %for the same $r$.
Interestingly, we show that this difference is a result of an inherent limitation of undirected strong amplifiers, and, in fact, our amplifier achieves the best possible fixation probability (modulo logarithmic factors) among all undirected graphs (for $r>1$ independent of $n$).

\begin{theorem}
    \label{thm:upperBound}
    For any undirected graph $G$ on $n$ vertices, $\rho_G(r) = 1 -  \Omega\left(\frac{1}{r^{5/3} n^{1/3} \log^{4/3}n}\right)$ for $r>1$.
\end{theorem}

The previous best bound was
$\rho_G(r) = 1-\Omega\big(n^{-2/3}\big)$, which assumed that $r$ is not a function of $n$~\cite{mertzios2013strong}.
The next theorem bounds the fixation probability for our strong suppressor.

\begin{theorem}
    \label{thm:suppressors-intro}
    There is a family of undirected graphs $\{B_n\}_{n\geq1}$, such that graph $B_{n}$ has $\Theta(n)$ vertices, and $\rho_{B_{n}}(r) = O\left(\frac{r^2\log n}{n^{1/4}}\right)$.
\end{theorem}

This statement implies that for $r$ bounded by a constant (independent of $n$), the fixation probability is $O(n^{-1/2}\log n)$.
Moreover, the fixation probability is $o(1)$ as long as $r = o(n^{1/8}\log^{-1/2} n)$; i.e., we have almost sure extinction even when the mutation has an \emph{extremely large} fitness, that is polynomial in $n$.

Both our strong amplifier and suppressor are surprisingly simple.
For example, the amplifier is as simple as follows:
All vertices have low degree, $\delta = 1$, except for a small number of hubs of degree $\Delta\approx n^{2/3}$; there are $n^{2/3} \log n$ such hubs.
The low-degree vertices form an independent set, while the hubs induce a random regular graph.
A subset of $n^{2/3}$ ``intermediary" hubs \gtodolight{structural holes ?}
are used to connect the low-degree vertices with the rest of the graph, such that each low-degree vertex is connected to a random such hub.
The simplicity of our graphs gives hope that they may indeed capture some aspects of real evolutionary networks, unlike previous strong amplifiers for which no structures resembling them have been reported in nature~\cite{jamieson2015fixation}.

The analysis of our strong suppressor is elementary.
Our analysis of the strong amplifier is more elaborate, using mainly Markov chain techniques, but is still intuitive and relatively short (in particular, significantly shorter than the analysis for directed strong amplifiers~\cite{galanis2015full}).
A tool of independent interest we introduce is a general lower bound on the probability that fixation will be reached in a given graph, as a function of the \emph{harmonic volume} (a.k.a.\ ``temperature''~\cite{lieberman2005evolutionary}) of the current set of mutant vertices (Lemma~\ref{lem:key}).
This result implies a convenient early stopping criterion that can be used to speed up numerical approximations of the fixation probability via simulations (Corollary~\ref{cor:hpfixation}).
Similar heuristics based on the number of mutants have already been used in the literature but without formal proof (e.g., in~\cite{barbosa2010early}).

Our proof of Theorem~\ref{thm:upperBound} relies on a careful argument that bounds the number of vertices of degree greater than a given $k$ as a function of the fixation probability of the graph, and also bounds the probability of fixation for a given initial mutant vertex in terms of the two-hop neighborhood of the vertex.

As already mentioned, a main motivation for our work comes from population genetics.
In this field, the effect of the spatial (or other) structure of a population of organisms on the dynamics of selection has long been a topic of interest~\cite{maruyama1974simple,slatkin1981fixation,barton1993probability,whitlock2003fixation}.
Similar evolutionary mechanisms apply also in somatic evolution within multicellular organisms.
For example, tissues are expected to be organized as suppressors of selection so as to prevent somatic evolution that leads to cancer~\cite{nowak2003linear,michor2004dynamics,komarova2006spatial,werner2011dynamics};
while amplifiers of selection can be found in the affinity maturation process of the immune response~\cite{lieberman2005evolutionary}.
The Moran process can be viewed also as a basic, general diffusion model, for instance, of the spread of ideas or influence in a social network~\cite{shakarian2015diffusion}.
In this context, our results suggest network structures which are well suited for the spread of favorable ideas, and structures where the selection of such ideas is suppressed.

We should stress that our results depend on the specifics of the Moran process, and in general do not extend to other similar diffusion processes, as discussed in the next paragraph.

\subsection{Other Related Work}

Several works have been devoted to the numerical computation of the fixation probability on specific classes of graphs, such as instances of bipartite graphs and Erd\"os-R\'enyi random graphs~\cite{broom2009evolutionarySmall,masuda2009directionality,%
voorhees2013fixation,hindersin2015most}.
%; see also survey~\cite{shakarian2012review}).

Besides the fixation probability, another quantity of interest is the \emph{absorption time}, that is, the expected time before either of the two absorbing states, fixation or extinction, is reached; a worst-case initial placement of the mutant is typically assumed.
D\'iaz et al.~\cite{diaz2014approximating} established polynomial upper bounds for the absorption time on undirected graphs.
They showed that this time is at most $\frac{n^3}{1-r}$ for $r<1$, at most $n^6$ for $r=1$, and at most $\frac{rn^4}{r-1}$ for $r>1$.
These results imply that there is a polynomial time approximation scheme for the fixation probability for $r\geq 1$ (for $r<1$ this probability can be exponentially small).
In~\cite{diaz2016absorption}, the bound for $r>1$  was refined for the case of regular graphs; the results apply also to directed strongly-connected regular graphs.
For these graphs the absorption time was show to be at least $\frac{(r-1)n\ln n}{r^2}$ and at most $\frac{2n\ln n}\phi$, where $\phi$ is the graph conductance.
%(for a definition of $\phi$, see Sec.~\ref{sec:prelims}).
It was also shown in~\cite{diaz2016absorption} that the absorption time can be exponential for non-regular directed graphs.

Several birth-and-death processes on graphs have been studied, beside the Moran process, which vary in the order in which the birth and death events occur in a step, and on whether the selection bias introduced by fitness is applied in the choice of the vertex that reproduces or the one that dies.
The Voter model~\cite{liggett1985interacting} is a well-known such process.
%Another variant is, e.g., when the vertex to die is chosen first, with probability \emph{inversely} proportional to its fitness, and is replaced by the offspring of a uniformly random neighbor (for $r=1$ this is just the well-known Voter model~\cite{liggett1985interacting}).
%Such processes have been studied extensively, and
It has been observe that small differences in these processes have significant impact on the fixation probability~\cite{antal2006evolutionary,hindersin2015most,kaveh2015duality}.
In particular, for dynamics where, unlike in Moran process, the vertex to die is chosen first, it is known that simple strong suppressors exits, such as the star $K_{1,n-1}$ with fixation probability $\Theta(1/n)$, while there are no strong amplifiers~\cite{meshkinfamfard2016randomised}.
Although not elaborated in the current paper, our results carry over to the process where the vertex to reproduce is chosen uniformly at random, and its offspring replaces a neighbor chosen with probability inversely proportional to its fitness.
This corresponds to a setting where advantageous mutants do not reproduce more often, but they live longer.

More involved population dynamics have been studied in evolutionary game
theory in the context of well-mixed populations~\cite{hofbauer1998evolutionary,gintis2000game}, and more recently in populations on graphs \cite{szabo2007evolutionary,shakarian2012review}.
The fitness of an individual is no longer fixed (determined by its type: mutant or non-mutant), but is determined by the expected payoff for the individual when playing a two-player game against a randomly chosen neighbor. In this game mutants play with one strategy and non-mutants with another.

Finally, similar stochastic processes have been use to model dynamics other than evolutionary ones, such as the spread of influence in social networks~\cite{kempe2003maximizing}, the spread of epidemics~\cite{hethcote2000mathematics}, the emergence of monopolies~\cite{berger2001dynamic}, or interacting particle systems~\cite{liggett1985interacting}.
%%\end{comment}

\paragraph{Road-map:} In Sec.~\ref{sec:prelims} we fix some definitions and notation.
We provide the description and analysis of our strong suppressor in Sec.~\ref{sec:suppressors}, and of our strong amplifier in Sec.~\ref{sec:amplifiers}.
In Sec.~\ref{sec:upperbound} we prove a universal upper bound on the fixation probability, and we conclude in Sec.~\ref{sec:conclusion}.

\section{Preliminaries}
\label{sec:prelims}

Let $G=(V,E)$ be an undirected graph with vertex set $V$ and edge set $E$, and let $n=|V|$.
For each vertex $u\in V$, $\Gamma_v$ denotes the set of vertices adjacent to $v$, and $\deg_u = |\Gamma_v|$ is the degree of $u$.
For any set $S\subseteq V$, the \emph{volume} of $S$ is $vol(S) = \sum_{u\in S}\deg_u$.
We define the \emph{harmonic volume} of $S$ as
\[
    hvol(S)=\sum_{u\in S}(\deg_u)^{-1}.
\]
The \emph{conductance} of $G$ is
\[
    \phi(G) = \min_{S\subseteq V\colon 0 < vol(S)\leq vol(V)/2}\frac{|E(S,V\setminus S)|}{vol(S)},
\]
where $E(S,V\setminus S)$ is the set of edges crossing the cut $\{S,V\setminus S\}$.
By $G[S]$ we denote the \emph{induced subgraph} of $G$ with vertex set $S$ and edge set consisting of all edges in $E$ with both endpoints in $S$.

The Moran process on $G=(V,E)$, with fitness parameter $r>0$, is a Markov chain $M_0,M_1,\ldots$ with state space $2^V$;
$M_t$ is the set of \emph{mutants} and $V\setminus M_t$ the set of \emph{non-mutants} after the first $t$ steps.
%The initial state is $M_0=\{s\}$, where $s$ is a vertex chosen uniformly at random from $V$.
In each step $t\geq 1$, if $M_{t-1} = S$ then:
(1)~a vertex $u\in V$ is chosen at random, such that $u$ is selected with probability
$\frac{r}{r|S| + (n-|S|)}$ if $u\in S$,
and
$\frac{1}{r|S| + (n-|S|)}$ if $u\notin S$; then
(2)~another vertex $v\in\Gamma_u$ is chosen uniformly at random; and
(3)~the chain moves to state $M_t = M_{t-1}\cup \{v\}$ if $u\in M_{t-1}$, or to $M_t = M_{t-1}\setminus \{v\}$ if $u\notin M_{t-1}$.
We say that in this step, $u$ \emph{reproduces} on $v$ and $v$ \emph{dies}.
It follows that for a given pair of adjacent vertices $u$ and $v$, if $M_{t-1}=S$ then $u$ reproduces on $v$ in step $t$ with probability $\frac{r/\deg_u}{r|S| + (n-|S|)}$ if $u\in M_{t-1}$, and $\frac{1/\deg_u}{r|S| + (n-|S|)}$ if $u\notin M_{t-1}$.
The states $V$ and $\emptyset$ are absorbing, and are called \emph{fixation} and \emph{extinction}, respectively.
By $\rho_G^S(r)$ we denote the probability that fixation is reached when the initial set of mutants is $M_0=S$, and by $\zeta_G^S(r)$ the probability that extinction is reached when $M_0=S$.
We let $\rho_G(r) = \frac{1}{n}\sum_{s\in V} \rho_G^{\{s\}}(r)$ denote the probability that fixation is reached when there is a single mutant initially, chosen uniformly among all vertices; we call $\rho_G(r)$ the \emph{fixation probability} of $G$.
Similarly, we let $\zeta_G(r) = \frac{1}{n}\sum_{s\in V} \zeta_G^{\{s\}}(r)$, and refer to is as the \emph{extinction probability} of $G$.
If $G$ is connected then $\rho_G^S(r) + \zeta_G^S(r) = 1$, and thus $\rho_G(r) + \zeta_G(r) = 1$.
In the following we will always assume a connected $G$.

%In the analysis, we will sometimes assume that initially some \emph{set} $S\subseteq V$ of vertices are mutants (rather than a single random vertex), and we will compute the probability that fixation is reached as a function of~$S$.

To simplify exposition we will often treat real numbers as integers, by implicity rounding their value up or down.
By $\log x$ we denote the logarithm base 2, while $\ln x$ denotes the natural
logarithm.

\section{Strong Suppressor}
\label{sec:suppressors}

\begin{figure}
\centering
\includegraphics[trim=0pt 90pt 0pt 100pt,clip=true, width=0.8\textwidth]{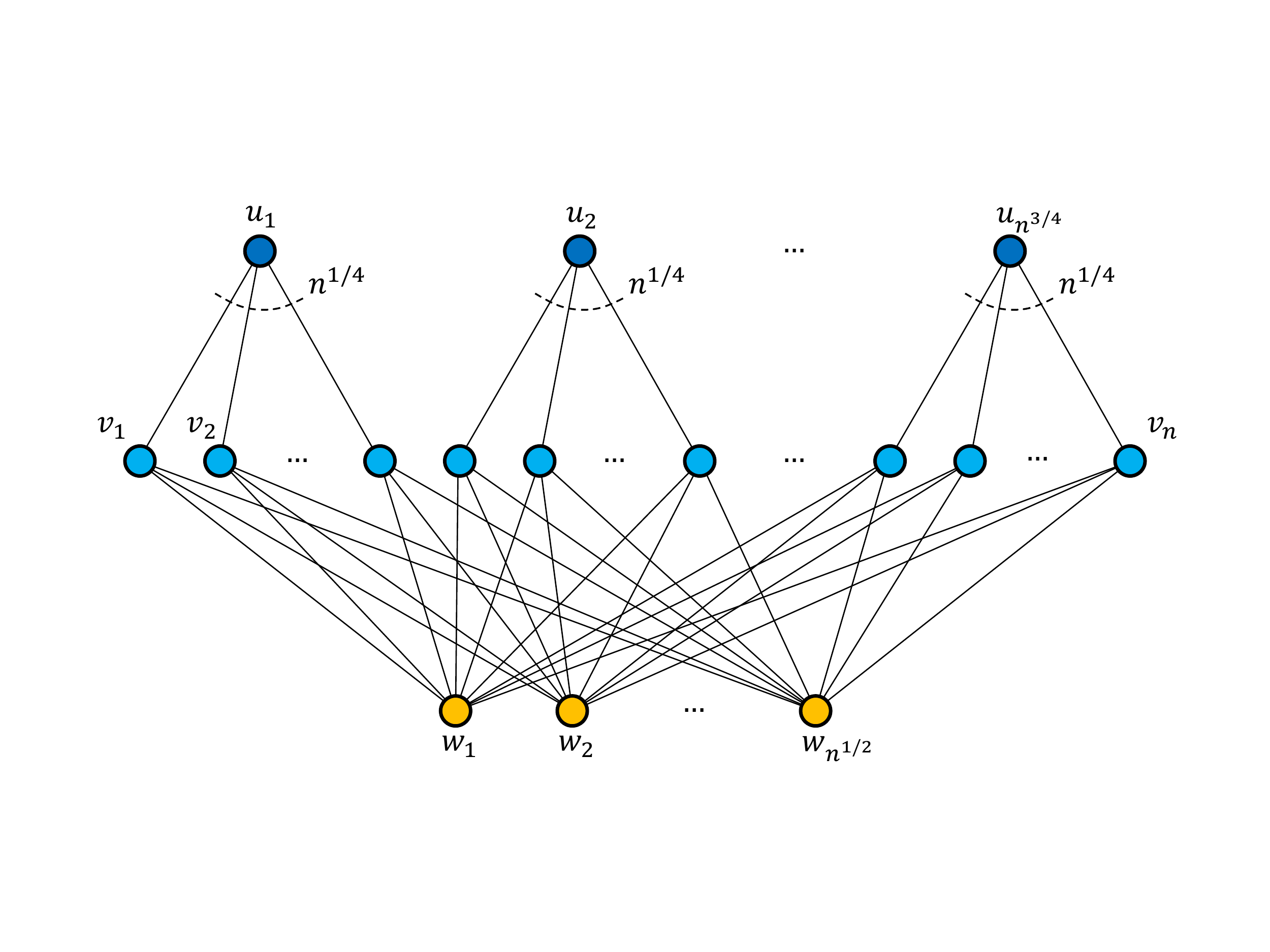}
\caption{Graph $B_n$ is a strong suppressor with fixation probability $O(r^2n^{-1/4}\log n)$.}
\label{fig:suppressor}
\end{figure}

Let $B_{n}$ be the following graph, illustrated in Fig.~\ref{fig:suppressor}.
The vertex set is the union of three disjoint sets of vertices $U$, $V$, and $W$, where $|U| = n^{3/4}$, $|V| = n$, and $|W| = n^{1/2}$.
So, the total number of vertices is $N = n + n^{3/4}+n^{1/4}$.
Edges exist only between vertices in $U$ and $V$ and between vertices in $V$ and $W$: each $u \in U$ is adjacent to $n^{1/4}$ distinct vertices in $V$, while each $v \in V$ is adjacent to a single vertex in $U$ and to all vertices in $W$ (i.e., the induced subgraph $B_n[V\cup W]$ is a complete bipartite graph $K_{n,n^{1/2}}$).

%\begin{itemize}
%  \item $|U| = n^{3/4}$ and each vertex $u \in U$ has degree $n^{1/4}$;
%  \item $|V| = n$ and each $v \in V$ has degree $1 + n^{1/2}$; and
%  \item $|W| = n^{1/2}$ and each $w \in W$ has degree $n$.
%\end{itemize}
%So, the total number of vertices is $N = n + n^{3/4}+n^{1/4}$.
%The edges are as follows:
%\begin{itemize}
%  \item each $u \in U$ is adjacent to $n^{1/4}$ vertices in $V$;
%  \item each $v \in V$ is adjacent to a single vertex $u\in U$, and to all vertices $w\in W$; and
%  \item each $w\in W$ is adjacent to all $v\in V$.
%\end{itemize}
%Graph $B_n$ is illustrated in Fig.~\ref{fig:suppressor}.
%\gtodo{mention: complete bipartite graph}

\begin{theorem}
    \label{thm:suppressors}
    The fixation probability of $B_{n}$ is
    $\rho_{B_n}(r) = O\left(\frac{r^2\log n}{n^{1/4}}\right)$.
\end{theorem}

The proof is fairly straightforward, and can be found in Sec.~\ref{sec:proofSuppressor}. Here we give a brief sketch.
For simplicity we assume  $r$ is bounded by a constant (independent of $n$), and argue that $\rho_{B_n}(r) = O(n^{-1/4}\log n)$.

The initial mutant $s$ belongs to set $V$ with probability $|V|/N = 1 - \Theta(n^{-1/4})$.
Suppose $s=v\in V$ and let $u$ be $v$'s neighbor in $U$.
Since $\deg_u\ll \deg_v$, it is highly likely that $u$ will reproduce on $v$ before $v$ reproduces on $u$ for the first time; the probability is $\Theta\big(\frac{1/\deg_u}{1/\deg_u+1/\deg_v}\big) = 1-\Theta(n^{-1/4})$.
The number of steps before $u$ reproduces on $v$ is at most $\lambda = \Theta(N\deg_u\log n)$ with probability at least $1-O(n^{-1/4})$.
In this interval, $v$ reproduces no more than $\mu = \Theta(\lambda/N) = \Theta(n^{1/4}\log n)$ times with probability at least $1-O(n^{-1/4})$.
An offspring of $v$ on a vertex $w\in W$ is, however, very unlikely to manage to reproduce before it dies, as $w$ has much higher degree than its $n$ neighbors in $V$; the probability is $\Theta\big(\frac1{1+n/\deg_v}\big)=\Theta(n^{-1/2})$.
So, the probability none of at most $\mu$ offspring of $v$ reproduces is $1-O(\mu n^{-1/2}) = 1-O(n^{-1/4}\log n)$.
Combining the above yields that with probability $1-O(n^{-1/4}\log n)$, the initial mutant is in $V$ and dies before reproducing on its neighbor in $U$, while all its offsprings in $W$ die before reproducing, as well; hence the mutation gets extinct.

We point out that $B_n$ is very different than previously known examples of (\emph{non-strong}) suppressors, in particular, the undirected suppressor proposed in~\cite{mertzios2013natural}, which consists of a clique $K_n$, an induced cycle $C_n$, and a matching between them.

\subsection{Proof of Theorem~\ref{thm:suppressors}}
\label{sec:proofSuppressor}

We compute a lower bound on the probability of the following event which implies extinction:
The initial mutant $s$ belongs to set $V$, the neighbor $u\in U$ of $s$ reproduces on $s$ before $s$ reproduces on $u$ for the first time, and until $u$ reproduces on $s$, whenever $s$ reproduces on a vertex $w\in W$, the offspring dies before it reproduces and also before $s$ has another offspring.

The event above implies that the total number of mutant vertices in the graph is never greater than~2.
This allows us to make the following technical modification to the random process when computing the probability of that event:
Whenever there are more than 2 mutants in the graph, the fitness of mutants drops from $r$ to 1.
This is convenient when we compute the probability by which each vertex is chosen to reproduce in a step, as it implies that the sum of the fitness of all vertices is always between $N$ and $(N-2) + 2r$.

Next we describe the precise events we consider, and compute their probability.
%Since the initial mutant $s$ is chosen uniformly at random among all vertices of $B_n$,
%\[
%    \Pr[s\in V]
%    =
%    \frac{|V|}N %{|U|+|V|+|W|}
%    =
%    \frac{n}{n+n^{3/4} + n^{1/4}} = 1-O(n^{-1/4}).
%\]
Suppose that the initial mutant $s$ is a given vertex $v\in V$.
Let $u$ denote the neighbor of $v$ in $U$.
Let $t_{u\to v}$ be the first step in which $u$ reproduces on $v$, and define $t_{v\to u}$ likewise.
In each step, the probability that  $u$ reproduces on $v$ is at least proportional to $1/\deg_u$, while the probability that  $v$ reproduces on $u$ is at most proportional to $r/\deg_v$.
It follows that
\[
    \Pr[t_{u\to v} < t_{v\to u}]
    \geq
    \frac{1/\deg_u}{ 1/\deg_u + r/\deg_v}
    =
    1 - O\left(r n^{-1/4}\right).
\]
The actual probability that $u$ reproduces on $v$ in a given step is at least $\frac{1/\deg_u}{N-2 + 2r}$ (here we have used the assumption that if there are more than two mutants in the system their fitness is 1).
It follows that for $\lambda = \frac14(N -2 + 2r)\deg_u\ln n$,
\[
    \Pr[t_{u\to v} > \lambda]
    \leq
    \left(1-\frac{1/\deg_u}{N -2+ 2r}\right)^{\lambda}
    \leq
    e^{- \frac{\lambda/\deg_u}{N -2+ 2r}}
    =
    n^{-1/4}.
\]
Let $Z$ denote the number times that $v$ reproduces in the first $\lambda$ steps.
Since the probability that $v$ reproduces in a given step is at most $\frac{r}{N-1+r}$, it follows that
\[
    \Exp[Z]
    \leq
    \frac{\lambda r}{N-1+r}
    =
    \frac{r(N -2 + 2r)\deg_u\ln n}{4(N-1+r)}
    \leq
    \frac{r\deg_u\ln n}{2}.
\]
Applying a standard Chernoff bound then yields that for $\mu = r \deg_u\ln n$,
\[
    \Pr\left[Z > \mu\right]
    <
    e^{-\Omega(\mu)}
    <
    n^{-1}.
\]
When $v$ reproduces on a neighbor $w\in W$, the probability that after this point, at least one of $v$ or $w$ reproduces before some vertex from $V\setminus\{v\}$ reproduces on $w$ is at most
$
    \frac{2r}{2r+(n-1)/\deg_v}\leq \frac{2r\deg_v}{n-1}:
$
in each step, one of $v,w$ is chosen to reproduce with probability at most proportional to $2r$ and each of the $n-1$ vertices  $v'\in V\setminus\{v\}$ is chosen to reproduce on $w$ with probability at least proportional to $1/\deg_{v'} = 1/\deg_v$.
Let $X_i$, $i\geq 1$, be the indicator variable that is 1 if the $i$-th time $v$ reproduces, it does so on some vertex $w\in W$, and then $v$ or offspring $w$ reproduce before a vertex from $V\setminus\{v\}$ reproduces on $w$.
As argued above
\[
    \Pr[X_i = 1] \leq \frac{2r\deg_v}{n-1}.
\]
We now combine the above, and use a union bound to establish a lower bound on the probability that all the following events occur simultaneously:
%(a)~$s$ belongs to $V$;
(a)~$u$ reproduces on $v$ for the first time before $v$ reproduces on $u$;
(b)~$u$ reproduces on $v$ within the first $\lambda$ steps;
(c)~in these $\lambda$ steps, $v$ reproduces at most $\mu$ times; and
(d)~in each of the first $\mu$ times that $v$ reproduces on $W$, the vertex $w$ on which it reproduces as well as $v$ are not chosen to reproduce before some other vertex reproduces on $w$.
The probability of this joint event is
\begin{align*}
    \Pr\Big[
%        s\in V
%        \ \wedge&\
        t_{u\to v} < t_{v\to u}
        \ \wedge&\
        t_{u\to v} \leq \lambda
        \ \wedge\
        Z \leq \mu
        \ \wedge\
        \textstyle \sum_{1\leq i\leq \mu} X_i = 0
    \Big]
    \\&
    \geq
    1
%    - O(n^{-1/4})
    - O\left(r n^{-1/4}\right) -  n^{-1/4} - n^{-1} - \frac{2r\mu \deg_v }{n-1}
    \\&
    =
    1-O(r^2n^{-1/4}\ln n)
    ,
\end{align*}
because
$
    \frac{2 r\mu \deg_v}{n-1}
    =
    \frac{2r^2 \deg_u \deg_v \ln n}{n-1}
    =
    O(r^2n^{-1/4}\ln n).
$
The above joint event implies that the mutation gets extinct before it spreads to any vertex in $U\cup V\setminus \{v\}$.

So far we have assumed that the initial mutant $s$ is a given $v\in V$.
Since $s$ is chosen uniformly at random among all vertices of $B_n$, we have $\Pr[s=v] = 1/N$.
Combining this with the result we showed above, it follows that the extinction probability (for a random initial mutant) is at least
\[
    \left(1-O(r^2n^{-1/4}\ln n)\right)\cdot \frac{|V|}{N}
    =
    \left(1-O(r^2n^{-1/4}\ln n)\right)\cdot \left(1-O(n^{-1/4})\right)
    =
    1-O(r^2n^{-1/4}\ln n).
\]
This implies the claim of Theorem~\ref{thm:suppressors}.
\qed

\section{Strong Amplifier}
\label{sec:amplifiers}

\begin{figure}
\centering
\includegraphics[trim=0pt 60pt 0pt 100pt,clip=true, width=0.8\textwidth]{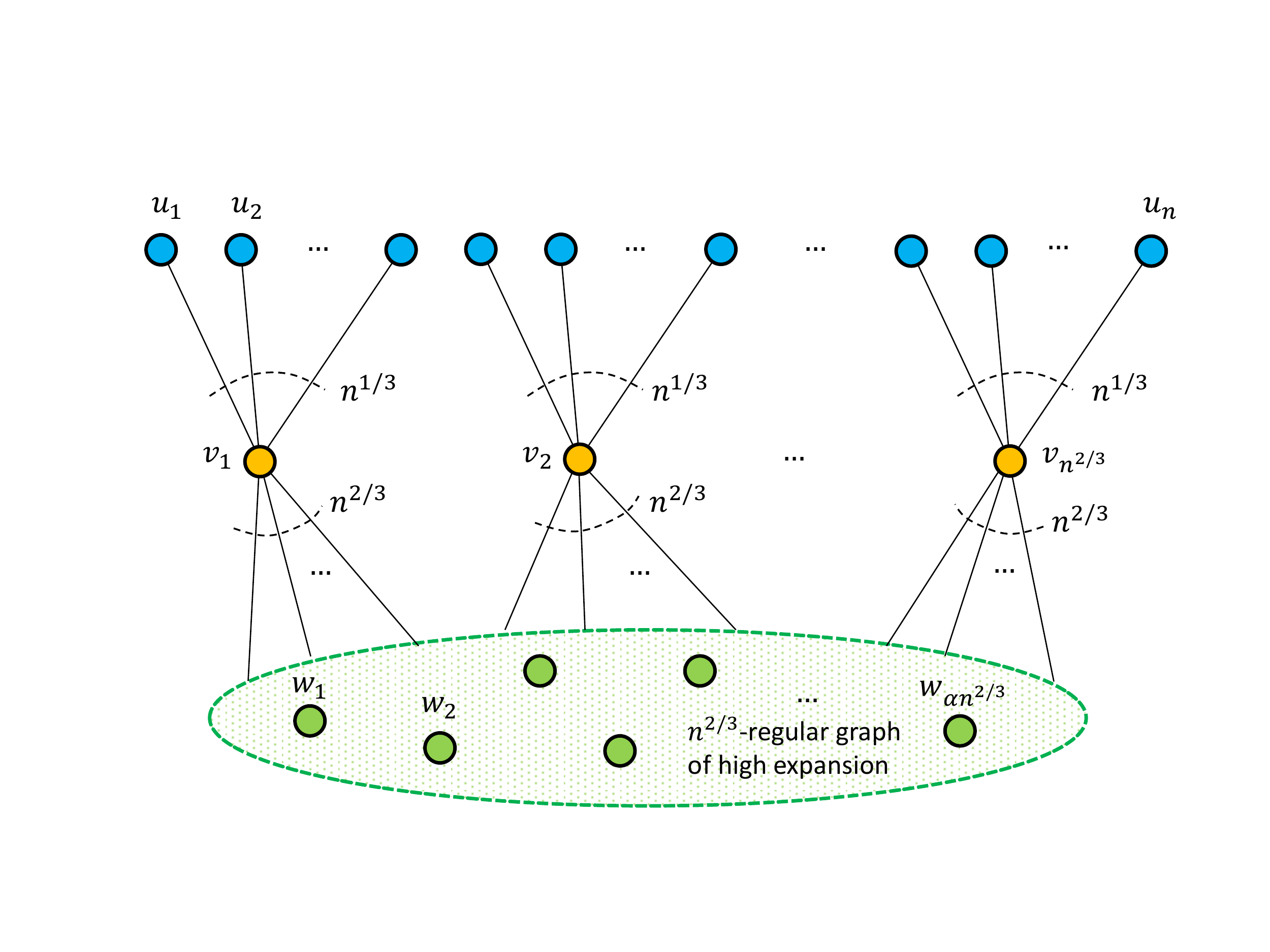}
\caption{Graph $A_{n,\epsilon}$ is a strong amplifier with fixation probability $1-O\big(\epsilon^{-1} n^{-1/3}\log n\big)$,
for $r \geq 1+\epsilon$.
We have $0<\epsilon\leq 1$ and $\alpha = 3\log_{1+\epsilon} n$.
The subgraph induced by $W = \{w_1,w_2,\ldots\}$ is an $n^{2/3}$-regular expander, and each vertex $w_i$ is also adjacent to $\alpha^{-1} n^{2/3}$ vertices in $V = \{v_1,v_2,\ldots\}$.
Although not shown in the figure, the edge set of the subgraph induced by $V$ can be arbitrary (in the figure, it is empty).}
\label{fig:amplifier}
\end{figure}

Let $A_{n,\epsilon}$, for $0<\epsilon \leq 1$, be the following graph, illustrated in Fig.~\ref{fig:amplifier}.
The vertex set is the union of three disjoint sets of vertices $U$, $V$, and $W$, where $|U| = n$, $|V| = n^{2/3}$, $|W| = \alpha\cdot n^{2/3}$ and
\[
    \alpha
    =
    3\log_{1+\epsilon} n.
\]
The total number of vertices is $N = n + (\alpha +1)n^{2/3}$, which is $\Theta(n)$ when $\epsilon = \Omega(n^{-1/3}\log n)$.
Each vertex $u \in U$ has degree 1 and is adjacent to a vertex in $V$.
Each $v \in V$ is adjacent to $n^{1/3}$ vertices in $U$ and $n^{2/3}$ vertices in $W$. Vertex $v$ may also be connected to other vertices in $V$; we do not impose any constraints on such connections, i.e., the edge set of the induced subgraph $A_{n,\epsilon}[V]$ can be arbitrary (in Fig.~\ref{fig:amplifier}, it is empty).
Finally, each $w\in W$ is adjacent to $\alpha^{-1}n^{2/3}$ vertices in $V$ and $n^{2/3}$ vertices in $W$.
The edges between vertices in $W$ are such that the $n^{2/3}$-regular induced subgraph $A_{n,\epsilon}[W]$ has high conductance.
For concreteness we assume
\[
    \phi(A_{n,\epsilon}[W]) \geq 1/3.
\]
We note that this requirement holds almost surely if $A_{n,\epsilon}[W]$ is a \emph{random} $n^{2/3}$-regular graph~\cite{bollobas1988isoperimetric}.

\begin{theorem}
    \label{thm:amplifiers}
    The fixation probability of $A_{n,\epsilon}$, for $\epsilon > 0$, is
    $\rho_{A_{n,\epsilon}}(r) = 1-O\left(\frac{\log n}{\epsilon n^{1/3}}\right)$, for $r \geq 1+\epsilon$.
\end{theorem}

Our proof of Theorem~\ref{thm:amplifiers} uses the following lemma, which is of independent interest.
The lemma establishes a lower bound on the probability $\rho_G^S$ that fixation is reached on a general graph $G$, as a function of the harmonic volume of the initial set $S$ of mutant vertices.
%Here we assume that initially an arbitrary set $S\subseteq V$ of vertices can be mutants, rather than a single random vertex.
Recall that the harmonic volume of vertex set $S$ is $hvol(S) = \sum_{v\in S}(\deg_v)^{-1}$.

\begin{lemma}
    \label{lem:key}
    Let $G = (V,E)$ be a connected graph with
    %$|V|= n$ vertices and
    minimum degree $\delta$,
    and let $S\subseteq V$ be an arbitrary set.
    The probability that fixation is reached on $G$ when the initial set of mutants is $S$ is $\rho_G^S(r) \geq \frac{1-r^{-\delta\cdot hvol(S)}}{1-r^{-\delta\cdot hvol(V)}}$.
\end{lemma}

Observe that for regular graphs, the lower bound predicted by Lemma~\ref{lem:key} matches the actual probability of fixation, $\frac{1-r^{-|S|}}{1-r^{-|V|}}$.
The proof of Lemma~\ref{lem:key}, in Sec.~\ref{sec:lem:key}, relies on the analysis of a variation of the gambler's ruin problem, in which the player can have an arbitrary strategy, subject to some constraints on the maximum bet amount, and the maximum profit per gamble.

In our proof of Theorem~\ref{thm:amplifiers}, we just use the following corollary of Lemma~\ref{lem:key}, which implies that if at some point the current set of mutant vertices has harmonic volume $\Omega(\delta^{-1}\log_r n)$, then fixation is subsequently reached with high probability.

\begin{corollary}
    \label{cor:hpfixation}
    Let $G = (V,E)$ be a connected graph with $|V|= n$ vertices and minimum degree $\delta$, and let $S\subseteq V$.
    If $hvol(S)\geq c\cdot \delta^{-1} \log_r n$, then $\rho_G^S(r) \geq 1-n^{-c}$.
\end{corollary}

%\begin{proof}
%From Lemma~\ref{lem:key}, fixation is reached with probability  at least $\frac{1-r^{-\delta\cdot hvol(S)}}{1-r^{-\delta\cdot hvol(V)}}$.
%Using that $hvol(S)\geq \frac{\log n}{\delta\log r}$ and $1-r^{-\delta\cdot hvol(V)}\leq 1$, we obtain that this probability is at least $1-1/n$.
%\end{proof}

With Corollary~\ref{cor:hpfixation} in hand, it suffices to show that with high probability the mutation in $A_{n,\epsilon}$ spreads to a set of vertices of sufficiently large harmonic volume, in order to prove that fixation is reached with high probability.
For that we use the next two lemmas.
The first establishes a lower bound of roughly $\frac{r-1}{2r}$
(for sufficiently large $\alpha$) on the probability that the mutation spreads from a single vertex in $W$ to half of the vertices in $W$ within $O\big(\frac{rN\log n}{r-1}\big)$ steps.
We will see that the harmonic volume of half of $W$ is large enough that allows us to apply Corollary~\ref{cor:hpfixation}.
In the following we denote by $M_t^S$ the set of mutants in $A_{n,\epsilon}$ after the first $t\geq 0$ steps, when the initial set of mutants is $S$.

\begin{lemma}
    \label{lem:expansion}
    Let $h_S = \min\{t\colon |M_t^S\cap W|\geq |W|/2\}$ be the number of steps before half of the vertices $w\in W$ are mutants, when the initial set of mutants is $S$.%
    \footnote{$h_S = \infty$ if the mutation is extinct before there is a point at which half of the vertices $w\in W$ are mutants.}
    For any set $S$ with $S\cap W\neq \emptyset$,
    \[
        \Pr\left[h_S \leq \frac{32Nr\log(\alpha n)}{r-1-4\alpha^{-1}} \right]
        \geq
%        \frac12-\frac1{2r} - \frac3{2r\alpha}
%        =
        \frac{r-1-4\alpha^{-1}}{2r}.
    \]
\end{lemma}

The intuition for this result is the following.
Since the induced subgraph $A_{n,\epsilon}$[W] is a regular expander, if we ignore all edges between vertices in $V$ and $W$, then from a single initial mutant in $W$ the mutation spreads to all of $W$ with probability $\frac{r-1}{r}$, as in any regular graph~\cite{lieberman2005evolutionary}; and if $r>1$ is not a function of $n$, this takes $O(n\log n)$ steps with high probability, because of the high expansion~\cite{diaz2016absorption}.
So, we need to account for the edges between vertices in $V$ and $W$, and for the fact that $r$ can be very close to 1.
It turns out that the effect of the edges between vertices in $V$ and $W$ is not significant as long as at most a constant fraction of the vertices in $W$ are mutants, because only a small, $\alpha^{-1}$, fraction of the neighbors of each $w\in W$ belong to $V$, and the degree of those neighbors is roughly the same as $w$'s.
The effect of small fitness values is that $h_S$ grows by roughly a factor of $(r-1)^{-1}$.
The proof of Lemma~\ref{lem:expansion}, in Sec.~\ref{sec:lem:expansion}, uses a coupling argument that relates the Moran process to a simpler birth-and-death Markov chain $\CC$ on $\{0,1,\ldots,|W|/2\}$, in which the birth probability at each step is larger than the death probability by a factor of roughly $r$.
The coupling is such that $\CC$ hits its maximum value of $|W|/2$ (before hitting 0) only after a point at which half of the vertices $w\in W$ are mutants in the Moran process.
This allows us to analyze $\CC$ instead.

The next lemma provides a lower bound of order $rn^{-1/3}$
on the probability that the mutation spreads from a vertex in $U$ to at least one vertex in $W$ within $O(N)$ steps.

\begin{lemma}
    \label{lem:spreadtoW}
    Let $o_S = \min\{t\colon |M_t^S\cap W|\neq\emptyset\}$ be the number of steps before at least one vertex $w\in W$ is a mutant, when the initial set of mutants is $S$.
    For any set $S$ with $S\cap U \neq \emptyset$,
    \[
        \Pr\left[o_S \leq 2N\right]
        \geq
        \frac{\big(1-O(n^{-1/3})\big)\,r}{18n^{1/3} + 6r}.
    \]
\end{lemma}

The proof of Lemma~\ref{lem:spreadtoW} is straightforward, and can be found in Sec.~\ref{sec:lem:spreadtoW}.
It uses similar type of arguments as the proof of Theorem~\ref{thm:suppressors}, to show that: (i)~with constant probability, the mutation spreads from a vertex $u\in U$ to its neighbor $v\in V$ within $N$ steps; and (ii) with probability
$\Omega\big(\frac{r}{n^{1/3} + r}\big)$, the mutation spreads from $v$ to some neighbor in $W$ within $N$ additional steps.

Before we prove the above lemmas, we show how we can use them to derive Theorem~\ref{thm:amplifiers}.
The detailed derivation can be found in Sec.~\ref{sec:proofAmplifier}.
Here, we just give some informal intuition.
For simplicity,%
\footnote{The intuitive arguments sketched here do not readily extend to the case where $\epsilon\to 0$ as $n$ grows; see Sec.~\ref{sec:proofAmplifier} for the full proof.}
we assume $\epsilon > 0$ is a constant (independent of $n$), and thus $\alpha=\Theta(\log n)$, and we argue $\rho_{A_{n,\epsilon}}(r) = 1- O(n^{-1/3}\log n)$.
We assume that $r=1+\epsilon$, as the result for $r>1+\epsilon$ follows from the case of $r=1+\epsilon$ by monotonicity.

The initial mutant $s$ belongs to set $U$ with probability $|U|/N = 1 - O(n^{-1/3}\log n)$.
Suppose that $s=u\in U$, and let $v$ be $u$'s neighbor in $V$.
Since $\deg_v > n^{2/3}$, it follows it will take $t = \Omega(N n^{1/3}\log n) = \Omega(n^{4/3}\log n)$ steps before $v$ reproduces on $u$, with probability $1-O(t/(N\deg_v)) = 1 - O(n^{-1/3}\log n)$.
Until then, $u$ will be a mutant and Lemma~\ref{lem:spreadtoW} implies that from any state in which $u$ is a mutant, it will take at most $2N$ steps before the mutation spreads to at least one vertex $w\in W$ with probability $\Omega(n^{-1/3})$.
If follows that within $2N/\Omega(n^{-1/3}) = O(Nn^{1/3})$ steps, at least one $w\in W$ will become a mutant with \emph{constant} probability.
Once this happens, Lemma~\ref{lem:expansion} implies that in another $O(N\log n)$ steps the mutation will spread to half of $W$ with constant probability.
Therefore, within $O(N n^{-1/3} + N\log n) = O(n^{4/3})$ steps in total, the mutation will spread from $u$ to half of $W$ with constant probability.
We can turn this constant into high probability by repeating the argument up to $\Theta(\log n)$ times, which would take $O(n^{4/3}\log n)$ steps.
Finally, once half of the vertices in $W$ are mutants, Corollary~\ref{cor:hpfixation} implies that fixation is subsequently reached with high probability, as the harmonic volume of half of $W$ is
$\frac{|W|/2}{\deg_w} > \log_{1+\epsilon} n$.
Combining all the above yields that fixation is reached with probability $1 - O(n^{-1/3}\log n)$.

\subsection{Proof of Theorem~\ref{thm:amplifiers} (Assuming Lemmas~\ref{lem:key},~\ref{lem:expansion}, and~\ref{lem:spreadtoW})}
\label{sec:proofAmplifier}

We will assume w.l.o.g.\ that $\epsilon = \omega(n^{-1/3}\log n)$, because if $\epsilon = O(n^{-1/3}\log n)$ the lower bound for the fixation probability predicted by the theorem becomes zero.
This assumption on $\epsilon$ implies that $\alpha = o(n^{1/3})$;
recall that
$
    \alpha = 3\log_{1+\epsilon} n = \Theta(\epsilon^{-1} \ln n).
$
We also observe that it suffice to assume that
\[
    r = 1+\epsilon,
\]
since the fixation probability is a monotone function of $r$~\cite{diaz2016absorption}, and thus $\rho_{A_{n,\epsilon}}(r) \geq \rho_{A_{n,\epsilon}}(1+\epsilon)$ for $r > 1+\epsilon$.

Next we fix some notation.
We consider the following collections of sets $S\subseteq U\cup V\cup W$:
\begin{gather*}
    \UU = \{S\colon S\cap U\neq\emptyset\},
    \qquad
    \OO = \{S\colon S\cap W\neq\emptyset\},
    \qquad
    \HH = \{S\colon |S\cap W|\geq|W|/2\},
    \\
    \SS_0 = \UU,
    \qquad
    \SS_1 = \UU \cap \OO,
    \qquad
    \SS_2 = \HH.
\end{gather*}
We define
\[
    \pi_i = \min_{S\in\SS_i} \rho_{A_{n,\epsilon}}^{S}(r),
    \qquad\text{for } i\in\{0,1,2\},
\]
i.e., $\pi_i$ is the probability that fixation is reached for the worst-case initial mutant set $S\in \SS_i$.
By $M_t$ we will denote the set of mutants after the first $t\geq0$ steps.
%when the initial mutant set is $M_0 = \{s\}$ for a uniformly random vertex $s$.
We are now ready to start the proof.

We first lower bound the fixation probability $\rho_{A_{n,\epsilon}}(r)$ in terms of $\pi_0$.
We have
\begin{equation}
    \label{eq:pifix}
    \rho_{A_{n,\epsilon}}(r)
    \geq
    \Pr[s\in U]\cdot\pi_0
    =
    \frac{|U|}{N}\cdot \pi_0
    =
    \frac{n\pi_0}{n+(\alpha+1)n^{2/3}}
    =
    \big(1-O(\alpha n^{-1/3})\big)\pi_0.
\end{equation}

Next we lower bound $\pi_0$ in terms of $\pi_1$.
Suppose that $M_t=S\in \SS_0$, $u\in S\cap U$, and $v\in V$ is the neighbor of $u$.
Lemma~\ref{lem:spreadtoW} implies that $M_{t+k}\in \OO$ for some (at least one) $k\in\{0,\ldots,2N\}$ with probability at least
\[
    p= \frac{\big(1-O(n^{-1/3})\big)\,r}{18n^{1/3} + 6r} = \Theta(n^{-1/3}),
\]
as $r = 1+\epsilon = \Theta(1)$.
Also the probability that $M_{t+k}\in \UU$ for \emph{all} $k\in\{0,\ldots,2N\}$ is greater or equal to the probability that $u$ is not replaced by a non-mutant in any of these $2N$ steps, which is at least $1-2N\cdot\frac{1/\deg_v}{N} \geq 1-2\varepsilon$, where
\[
    \varepsilon = \max_{v'\in V} \{1/\deg_{v'}\} \leq n^{-2/3}.
\]
So, if $M_t=S\in \SS_0$, then:
(i)~the probability that $M_{t+k} \in \OO \cap \UU = \SS_1$ for some $k\in\{0,\ldots,2N\}$ is at least $p-2\varepsilon$ (by a union bound), and
(ii)~the probability that $M_{t+2N}\in \SS_0$ is at least $1-2\varepsilon$.
It follows
\begin{align*}
    \pi_0
    &\geq
    (p-2\varepsilon) \cdot \pi_1
    +
    [(1-2\varepsilon) - (p-2\varepsilon)] \cdot \pi_0.
%    (1-2\varepsilon)\cdot\pi_0 +
%    (p-2\varepsilon)\cdot(\pi_1-\pi_0).
\end{align*}
Solving for $\pi_0$ yields
\begin{equation}
    \label{eq:pi0}
    \pi_0
    \geq
    \left(1-2\varepsilon/p\right)\pi_1
    =
    \big(1-O(n^{-1/3})\big)\pi_1
    .
\end{equation}

In a similar manner, we lower bound $\pi_1$ in terms of $\pi_0$ and $\pi_2$.
Suppose that $M_t=S\in \SS_1$, $u\in S\cap U$, and $v\in V$ is the neighbor of $u$.
Lemma~\ref{lem:expansion} implies that $M_{t+k}\in \HH$ for some $k\in\{0,\ldots,\ell N\}$, where
\[
    \ell
    =
%    \frac{32r\ln(\alpha n)}{r-1-4\alpha^{-1}}
%    =
    \frac{32(1+\epsilon)\log(\alpha n)}{\epsilon-4\alpha^{-1}}
    =
    \Theta(\epsilon^{-1}\log n)
    =
    O(n^{1/3}),
\]
with probability at least
\[
    q
%    = \frac{r-1-4\alpha^{-1}}{2r}
    = \frac{\epsilon-4\alpha^{-1}}{2(1+\epsilon)}
    =
    \Theta(\epsilon)
    ,
\]
as $\alpha =\Theta(\epsilon^{-1}\log n)$ and $\epsilon = \Omega(n^{-1/3}\log n)$.
The probability that $M_{t+k}\in\UU$ for \emph{all} $k\in\{0,\ldots,\ell N\}$ is at least equal to $1-\ell N\cdot\frac{1/\deg_v}{N} \geq 1-\ell\varepsilon$.
So, if $M_t=S\in \SS_1$, then:
(i)~the probability that $M_{t+k} \in \SS_2 $ for some $k\in\{0,\ldots,\ell N\}$ is at least $q$, and
(ii)~the probability that $M_{t+\ell N}\in \SS_0$ is at least $1-\ell\varepsilon$.
This implies
\begin{equation}
    \label{eq:pi1}
    \pi_1
    \geq
    q\cdot \pi_2 +
    (1-\ell\varepsilon - q)\cdot\pi_0
    =
    q\cdot \pi_2 +
    \big(1- q - O(n^{-1/3})\big)\cdot\pi_0
    .
\end{equation}

Finally, we lower bound $\pi_2$ using Corollary~\ref{cor:hpfixation}.
If $M_t=S\in \SS_2$  and $w\in W$ then
\[
    hvol(S)
    \geq
    \frac{|W|/2}{\deg_w}
    =
    \frac{\alpha/2}{1 + \alpha^{-1}}
    >
    \frac{\alpha}{3}
    =
    \log_{1+\epsilon} n
    >
    \frac12\log_{1+\epsilon} N,
\]
as $\alpha = 3\log_{1+\epsilon}n > 3$ and $n>N^{1/2}$.
Applying Corollary~\ref{cor:hpfixation} (for $r=1+\epsilon$, $\delta = 1$, and $c=1/2$) yields
\begin{equation}
    \label{eq:pi2}
    \pi_2
    \geq
    1-N^{-1/2}.
\end{equation}

We now combine~\eqref{eq:pifix}--\eqref{eq:pi2}.
From~\eqref{eq:pi0} and~\eqref{eq:pi1},
\[
    \pi_1
    \geq
    \frac{q\pi_2}{1-(1-q - O(n^{-1/3}))(1-O(n^{-1/3}))}
    =
    \big(1-O(q^{-1}n^{-1/3})\big)\pi_2
    .
\]
From this,~\eqref{eq:pifix},~\eqref{eq:pi0}, and~\eqref{eq:pi2},
\begin{align*}
    \rho_{A_{n,\epsilon}}(r)
    &=
    \big(1-O(\alpha n^{-1/3})\big)
    \cdot
    \big(1-O(n^{-1/3})\big)
    \cdot
    \big(1-O(q^{-1}n^{-1/3})\big)
    \cdot
    \big(1-N^{-1/2}\big)
    \\&
    =
    1-O\big((\alpha+q^{-1})n^{-1/3}\big)
%    \\&
    =
    1-O\left(\epsilon^{-1}\log n \cdot n^{-1/3}\right).
\end{align*}
This completes the proof of Theorem~\ref{thm:amplifiers}.
\qed

\subsection{Proof of Lemma~\ref{lem:spreadtoW}: Spread of Mutation to a Vertex in \texorpdfstring{$W$}{W}}
\label{sec:lem:spreadtoW}

We show that: (i)~with constant probability, the mutation spreads from a vertex $u\in U$ to its neighbor $v\in V$ within $N$ steps; and (ii) with probability
$\Omega\big(\frac{r}{n^{1/3} + r}\big)$, the mutation spreads from $v$ to some neighbor in $W$ within $N$ additional steps.

Suppose $u\in S\cap U$, and let $v$ be the neighbor of $u$ in $V$.
Let $t_v = \min\{t\colon v\in M_t^{S}\}$ be the first time when $v$ is a mutant.
We show
\begin{equation}
    \label{eq:tv}
    \Pr[t_v \leq N]
    \geq
    1-e^{-1} -n^{-2/3}.
\end{equation}
Let $t_{u\to v}$ be the first step in which $u$ reproduces on $v$, and define $t_{v\to u}$ similarly.
Let $t_1 = \min\{t_{u\to v}, t_{v\to u}\}$.
In any step $i\leq t_1$, the probability $u$ is chosen for reproduction is at least $\frac{r}{r N}=1/N$, since $u$ is a mutant.
It follows
\[
    \Pr[t_1\leq N] \geq 1 - \left(1-1/N\right)^N \geq 1-e^{-1}.
\]
Also the probability that $u$ reproduces on $v$ before $v$ reproduces on $u$ is
\[
    \Pr[t_{u\to v} = t_1]
    \geq
    \frac{r}{r + r/\deg_v}
%    =
%    \frac{\deg_v}{\deg_v + 1}
%    =
%    1 - \frac{1}{\deg_v + 1}
    \geq
    1 - (\deg_v)^{-1}
    \geq
    1- n^{-2/3},
\]
where the first inequality holds because the probability $u$ reproduces (on $v$) in a step is proportional to $r$ while the probability $u$ reproduces on $v$ is at most proportional to $r/\deg_v$.
From the last two results above and a union bound, we get
\[
    \Pr[t_{u\to v} = t_1 \leq N]
    \geq
    1-e^{-1} -n^{-2/3}.
\]
Observing that the event $t_{u\to v} = t_1 \leq N$ on the left implies that $v$ is a mutant right after step $t_{u\to v} \leq N$, and thus $t_v\leq N$,
we obtain~\eqref{eq:tv}.

Next we show that
\begin{equation}
    \label{eq:os}
    \Pr[o_S \leq k + N \mid t_v=k]
    \geq
    \frac{(1-e^{-1/3})\, r}{3n^{1/3}+ r + 6}.
\end{equation}
Let $t_{v\to W}$ be the first step after step $t_v = k$ in which $u$ reproduces on some neighbor in $W$, and let $t'_{v}$ be the first step after $k$ in which a non-mutant neighbor of $v$ reproduces on $v$ (so, $v$ is a mutant until step  $t'_v$).
Let $t_2 = \min\{t_{v\to W}, t'_{v}\}$.
In any step $k<i\leq t_2$, the probability that $v$ reproduces on some neighbor in $W$ is at least
$
\frac{r n^{2/3}/\deg_v}{r N}
\geq
\frac{1}{3N}
.$
It follows
\[
    \Pr[t_2 \leq k + N \mid t_v = k] \geq 1 - \big(1 - 1/(3N)\big)^N\geq 1-e^{-1/3}.
\]
Further, given $t_2 \leq k +N$, the probability that in step $t_2$ vertex $v$ reproduces on some neighbor in $W$ is proportional to $r n^{2/3}/\deg_v \geq r/3$ while the probability that some non-mutant neighbor reproduces on $v$ is at most proportional to $n^{1/3} + \sum_{v'\in V}(1/\deg_{v'}) + \sum_{w\in W\cap \Gamma_v}(1/\deg_{w}) \leq n^{1/3} + 2$.
Hence,
\[
    \Pr[t_{v\to W} = t_2 \mid t_2 \leq k +N, t_v = k]
    \geq
    \frac{r/3}{r/3 + n^{1/3} + 2}
    \geq
    \frac{r}{r + 3n^{1/3} + 6}.
\]
Combining the last two results above yields
\[
    \Pr[t_{v\to W} = t_2 \leq k + N \mid t_v =  k]
    \geq
    \frac{(1-e^{-1/3})\, r}{r + 3n^{1/3} + 6}.
\]
Observing that the event $t_{v\to W} = t_2 \leq k + N$ implies $o_S \leq k + N$ completes the proof of~\eqref{eq:os}.

Finally, combining~\eqref{eq:tv} and~\eqref{eq:os} we obtain
\begin{align*}
    \Pr[o_S\leq 2N]
    &\geq
    \sum_{k\leq N}
    \big(
    \Pr[o_S \leq k + N \mid t_v=k]
    \cdot
    \Pr[t_v = k]
    \big)
    \geq
    \frac{(1-e^{-1/3})\, r}{3n^{1/3} + r + 6}
    \cdot
    \sum_{k\leq N}
    \Pr[t_v = k]
    \\&
    \geq
    \frac{(1-e^{-1/3})\, r}{3n^{1/3} + r + 6}
    \cdot
    \big(1-e^{-1} -n^{-2/3}\big)
    \geq
    \frac{\big(1-O(n^{-1/3})\big)\,r}{18n^{1/3} + 6r}.
\end{align*}
This completes the proof of Lemma~\ref{lem:spreadtoW}.\qed

\subsection{Proof of Lemma~\ref{lem:expansion}: Spread of Mutation to Half of \texorpdfstring{$W$}{W}}
\label{sec:lem:expansion}

We devise a coupling of the Moran process on $A_{n,\epsilon}$ and an appropriate birth-and-death Markov chain $\CC$ on $\{0,1,\ldots,|W|/2\}$.
This coupling ensures that $\CC$ hits its maximum value of $|W|/2$ (before hitting 0) only after half of the vertices $w\in W$ are mutants in the Moran process.
This allows us to analyze $\CC$ instead.

Let $M_t$ denote the set of mutants in $A_{n,\epsilon}$ after the first $t\geq0$ steps of the Moran process, and let $W_t = M_{t}\cap W$ be the set of mutants in $W$.
We assume that $M_0 = S$ and $W_0 = S\cap W \neq \emptyset$.
We also assume that $|W_0|<|W|/2$, otherwise $h_S=0$ and the lemma holds trivially.
We will use the shorthand notation
\[
    \kappa = |W|/2, \qquad k_0 = |W_0| \in\{1,\ldots,\kappa-1\}.
\]
Recall that for any vertex $w\in W$, $\deg_w = (1+\alpha^{-1})n^{2/3}$, and for any vertex $v\in V$, $\deg_v \geq n^{2/3}+n^{1/3}$.
Let $e_t$ be the number of edges between $W_t$ and $W\setminus W_t$.
Since $A_{n,\epsilon}[W]$ is a $n^{2/3}$-regular graph and its conductance is at least $1/3$,
\[
    e_t \geq
    n^{2/3}|W_t|/3,
    \text{\qquad if $|W_t| \leq \kappa$}.
\]
Let $e'_t$ be the number of edges between $W_t$ and $V$.
We have $e'_t = \alpha^{-1}n^{2/3}|W_t|$, and thus
\[
    e'_t
    \leq
    3\alpha^{-1}e_t,
    \text{\qquad if $|W_t| \leq \kappa$}.
\]

Let $p_T$ denote the probability that the number of mutants in $W$ increases in the next step of the Moran process, when the current set of mutants is $T$, and let $q_T$ be the probability that the number of mutants in $W$ decreases instead.
For any set $T$, if $|T| = \lambda$ and $|T\cap W| = k\in\{1,\ldots,\kappa-1\}$,
\begin{align}
    \label{eq:pk-moran}
    p_T
    =
    \Pr\big[
    |W_{t+1}| = k+1
    \,\big|\,
    M_{t} = T
    \big]
    &\geq
    \frac{r e_t/\deg_w}{N-\lambda + \lambda r}
    \geq
    \frac{k}{4N},
    \\
    \label{eq:qk-moran}
    q_T
    =
    \Pr\big[
    |W_{t+1}| = k-1
    \, \big|\,
    M_{t} = T
    \big]
    &\leq
    \frac{e_t/\deg_w}{N-\lambda + \lambda r}
    + \frac{e'_t/\min_{v\in V}\deg_v}{N-\lambda + \lambda r}
%    \notag\\
    \leq
    p_T\!\cdot\! \frac{1+4\alpha^{-1}}{r},
\end{align}
where for the last inequality in~\eqref{eq:pk-moran} we used that
$e_t \geq n^{2/3}|W_t|/3 = n^{2/3}k/3$,
$N-\lambda + \lambda r\leq rN$,
and
$\deg_w = (1+\alpha^{-1})n^{2/3}\leq (4/3)n^{2/3}$, as $\alpha\geq 3$;
whereas for the last inequality in~\eqref{eq:qk-moran} we used that $e'_t \leq 3e_t \alpha^{-1}$, and that for any $v\in V$, $\deg_w/\deg_v\leq 1+\alpha^{-1}\leq 4/3$.
We will assume that
\[
    \frac{1+4\alpha^{-1}}{r} < 1,
\]
because otherwise the lemma holds trivially.
Hence, $p_T > q_T$.

We define the birth-and-death Markov chain $\CC = (C_0,C_1,\ldots)$ as follows.
The state space is $\{0,\ldots,\kappa\}$; the initial state is $C_0 = k_0$; there is one absorbing state, 0; and the birth/death probabilities at state $k$ are
\begin{align*}
    b_k
    =
    \Pr[C_{t+1} = k+1\mid C_t=k]
    &=
    \frac{k}{8N},
    \qquad
    \text{for } 0< k<\kappa,
\\
    d_k
    =
    \Pr[C_{t+1} = k-1\mid C_t=k]
    &=
    b_k\cdot
    \frac{1+4\alpha^{-1}}{r},
    \qquad
    \text{for } 0< k\leq\kappa.
\end{align*}
Comparing to the corresponding probabilities~\eqref{eq:pk-moran} and~\eqref{eq:qk-moran}, we have that $b_k \leq p_T/2$ and $d_k/b_k \geq q_{T}/p_{T}$.
Intuitively, this means that in $\CC$, transitions towards $\kappa$ are less likely than in the Moran process, while the relative rate of transitions towards zero is higher.
The next claim formalizes a consequence of this intuition.

\begin{claim}
    \label{clm:coupling}
    There is a coupling of processes $\MM  = (M_0, M_{1},\ldots)$ and $\CC = (C_0,C_1,\ldots)$ such that for any $t\geq 0$, it holds\footnote{With a slight abuse of notation, we will use the same symbols, $\MM$ and $\CC$, to denote the coupled copies of the two processes.}
    $C_t \leq \max_{0\leq j\leq t}\{|W_j|\}$.
\end{claim}

The above claim implies that if $C_t = \kappa$, then $|W_j|\geq \kappa$ for some $j\leq t$, i.e., $h_S \leq t$.
Hence, a lower bound on the probability that $\CC$ hits $\kappa$ within $t$ steps, is also a lower bound on the probability that $h_s \leq t$.
The next claim provides such a lower bound for $\CC$.

\begin{claim}
    \label{clm:hitting}
    With probability at least $\frac{r-1-4\alpha^{-1}}{2r}$,
    $\CC$ hits $\kappa$ in at most
    $\frac{32Nr\log(\alpha n)}{r-1-4\alpha^{-1}}$ steps.
\end{claim}

From Claim~\ref{clm:coupling} it follows that the probability of having $h_S \leq t= \frac{32Nr\log(\alpha n)}{r-1-4\alpha^{-1}}$ is greater or equal to the probability that $\CC$ hits $\kappa$ within $t$ steps, and the latter probability is at least $\frac{r-1-4\alpha^{-1}}{2r}$, by Claim~\ref{clm:hitting}.
This proves Lemma~\ref{lem:expansion}.
It remains to prove Claims~\ref{clm:coupling} and~\ref{clm:hitting}.

\paragraph{Proof of Claim~\ref{clm:coupling}: The Coupling of $\MM$ and $\CC$.}
%\label{sec:proofClaimCoupling}

We just need to specify the coupling of the two processes for as long as both conditions $C_t > 0$ and $|W_t| < \kappa$ hold.
The reason is that if $C_t = 0$ or $|W_t| \geq \kappa$, then for all $t'\geq t$ we have respectively $C_{t'} = 0 < 1 \leq \max_{0\leq j\leq t'}\{|W_j|\}$ or $C_{t'} \leq \kappa \leq  \max_{0\leq j\leq t}\{|W_j|\}$, regardless of the coupling; hence, we can let the two processes take steps independently %of each other
from that point on.

The coupling is as follows.
As long as the value of $\CC$ is smaller than the number $k$ of mutants in $W$, we let $\CC$ take steps (independently of $\MM$), until it hits $k$ or the absorbing state 0; we refer to these steps as \emph{catch-up} steps.
When the state of $\CC$ is the same as the number $k$ of mutants in $W$, we first let $\MM$ do a step.
Suppose this is the $(t+1)$-th step of $\MM$, which moves from state $M_t$ to $M_{t+1}$, and suppose $M_{t} = T$; also, as already mentioned, $|W_t| = k\in\{1,\ldots,\kappa-1\}$.
Then we let $\CC$ do one or more steps as follows, depending on the step of $\MM$.

\begin{itemize}
  \item
    If $|W_{t+1}| = |W_t| = k$ then $\CC$ does just one step; this step does not change the state of $\CC$, which remains in state $k$.
  \item
    If $|W_{t+1}| = k+1$ then $\CC$ first does zero or more steps in which it remains in state $k$, and then moves to either $k+1$ or $k-1$.
    Precisely, we perform a sequence of independent coin tosses until a head comes up.
    The first coin has probability of  heads $\frac{b_k+d_k}{p_T+q_T}$, and the rest $b_k+d_k$.%
    \footnote{Note that both quantities $\frac{b_k+d_k}{p_T+q_T}$ are $b_k+d_k$ are at most 1, because $d_k\leq b_k \leq p_T/2$.}
    Chain $\CC$ does one step for each coin toss.
    For each tail, it remains at state $k$, while for the head, it moves to $k+1$ with probability
    \[
        x_k = \frac{b_k(p_T+q_T)}{p_T(b_k+d_k)},
    \]
    and to $k-1$ with probability $1-x_k$.
    (Note that $x_k\leq 1$, because $q_T/p_T\leq d_k/b_k$.)

  \item
    If $|W_{t+1}| = k-1$ then $\CC$ first does zero or more steps in which it remains in state $k$, and then moves to state $k-1$.
    As in the previous case,  we perform a sequence of coin tosses until a head comes up, where the first coin has probability of heads $\frac{b_k+d_k}{p_T+q_T}$, and the rest $b_k+d_k$.
    For each tail, $\CC$ does a step in which it remains at $k$, while for the head it moves to $k-1$.
\end{itemize}

Before we show that the above coupling is valid, i.e., it yields the correct marginal distribution for the processes, we argue that it satisfies the desired inequality $C_t \leq \max_{0\leq j\leq t}\{|W_j|\}$.
Let $\sigma_t$, for $t\geq0$, denote the number of steps that $\CC$ does (from the beginning) until $\MM$ does its $(t+1)$-th step.
We argue by induction that for $0\leq t\leq h_S$,
\[
    C_{\sigma_t} \in\{|W_t|,0\}
    \quad\text{and}\quad
    C_{\sigma_t} \leq \max_{0\leq j\leq t}\{|W_j|\}.
\]
These two equations hold for $t=0$, because  $C_0 = k_0= |W_0|$ and $\sigma_0 = 0$.
Moreover, if they hold for some $t\geq 0$, for which $|W_t|=k<\kappa$ and $C_t>0$, then the coupling rules ensures that in its next $\sigma_{t+1}-\sigma_t$ steps, $\CC$ will not exceed $k+1$ and will hit $k+1$ or $0$ at step $\sigma_{t+1}$, if $|W_t+1| = k+1$; while if $|W_t+1| = k-1$, $\CC$ will stay at $k$ and only move to $k-1$ at step $\sigma_{t+1}$.
Thus the desired inequalities hold also for $t+1$, completing the inductive proof.

It remains to argue that the coupling is valid.
Namely, we must show that the transitions in $\CC$ have the correct distribution.
By definition, the catch-up steps by $\CC$ have the right transition probabilities, so we just need to argue about the remaining steps.
%sequence of steps of $\CC$ that correspond to each step of $\MM$.

Suppose $\CC$ is at state $0<k<\kappa$, $\MM$ is at state $T$ with $k$ mutants in $W$, and it is $\MM$'s turn to do a step.
The probability that $\CC$ moves to $k+1$ next and the transition occurs after $j\geq 1$ steps (in the first $j-1$ steps, the chain remains at state $k$) must be
\[
    (1-b_k-d_k)^{j-1}b_k.
\]
The corresponding probability our coupling yields equals the probability that all the following happen: (a)~for some $l<j$, we have $l$ consecutive steps of $\MM$ in which the number of mutants in $W$ remains $k$, and thus we have the same number of steps by $\CC$ in which it remains at $k$; (b)~the number of mutants increases to $k+1$ in the next step of $\MM$; (c)~we then perform $j-l$ coin tosses for $\CC$ before a head comes up, as described in the coupling; and, last, (d)~we do a successful Bernoulli trial with success probability  $x_k$.
The probability of the above joint event is
\begin{align}
    \label{eq:big-prob}
    &\sum_{l=0}^{j-2}
    (1-p_T-q_T)^l
    p_T
    \left(1-\frac{b_k+d_k}{p_T+q_T}\right)
    (1-b_k-d_k)^{j-l-2}
    (b_k+d_k)
    x_k
    \notag\\
    +\ &
    (1-p_T-q_T)^{j-1}
    p_T
    \frac{b_k+d_k}{p_T+q_T}
    x_k.
\end{align}
For $j=1$, the sum in the first line is empty, and the term  in the second line is $b_k$ because $x_k = \frac{b_k(p_T+q_T)}{p_T(b_k+d_k)}$.
So, in this case the probability above equals $(1-b_k-d_k)^{j-1}b_k$, as required.
For $j> 1$, the sum in the first line of~\eqref{eq:big-prob} becomes as follows after substituting $x_k$'s value and rearranging, and then computing the resulting sum:
\begin{align*}
    &
    (1-b_k-d_k)^{j-1}b_k \cdot\sum_{l=0}^{j-2}
    \left(\frac{1-p_T-q_T}{1-b_k-d_k}\right)^l
    \left(1-\frac{b_k+d_k}{p_T+q_T}\right)
    (1-b_k-d_k)^{-1}
    (p_T+q_T)
    \\
    =\ &
    (1-b_k-d_k)^{j-1}b_k \cdot
    \left(1 - \left(\frac{1-p_T-q_T}{1-b_k-d_k}\right)^{j-1}\right).
\end{align*}
The term in the second line of~\eqref{eq:big-prob} is
\[
    (1-p_T-q_T)^{j-1}
    p_T
    \frac{b_k+d_k}{p_T+q_T}
    x_k
    =
    (1-b_k-d_k)^{j-1}b_k
    \cdot
    \left(\frac{1-p_T-q_T}{1-b_k-d_k}\right)^{j-1}.
\]
Combining the above yields that the probability in~\eqref{eq:big-prob} equals $(1-b_k-d_k)^{j-1}b_k$, as desired.

In a similar fashion, we show that the correct probability of
\[
    (1-b_k-d_k)^{j-1}d_k
\]
is ensured by our coupling for $\CC$ moving from state $k$ to $k-1$ after $j\geq 1$ steps (in the first $j-1$ steps the chains remains at $k$).
The expression for this probability is identical to~\eqref{eq:big-prob}, except that in each term of the sum in the first line, and also in the term in the second line, the factor $p_T x_k$ is replaced by $q_T + p_T(1-x_k)$.
By observing that
\[
%    p_T x_k = \frac{b_k(p_T+q_T)}{b_k+d_k}
%    \qquad
    q_T + p_T(1-x_k)
%    =
%    q_T + p_T - \frac{b_k(p_T+q_T)}{b_k+d_k}
%    =
%    \frac{(p_T + q_T)d_k}{b_k+d_k}
    =
    p_Tx_k\cdot \frac{d_k}{b_k},
\]
the desired result follows immediately from the result for the transition from state $k$ to $k+1$.
This completes the argument that the coupling is valid, and also the proof of Claim~\ref{clm:coupling}.
%\end{proof}
\qed

\paragraph{Proof of Claim~\ref{clm:hitting}: Analysis of $\CC$.}
%\label{sec:proofClaimHitting}

Recall that the birth and death probabilities in $\CC$ are respectively
\[
    b_k
    =
    bk
    \quad\text{and}\quad
    d_k
    =
    \gamma bk,
    \quad\text{where }
    b = 1/(8N)
    \quad\text{and}\quad
    \gamma
    =
    (1+4\alpha^{-1})/r < 1
    .
\]
First we compute the probability that $\CC$ hits $\kappa$ (before hitting the absorbing state 0), without concern about time.
This probability is the same as the probability that a biased random walk on the line hits $\kappa$ before $0$, when the walk starts at $k_0$, and the transition probabilities associated with moving right and left are respectively
\[
    p
    =
    \frac{b_k}{b_k+d_k}
    =
    \frac{1}{1+\gamma},
    \qquad
    1-p =
%    \frac{d_k}{b_k+d_k}
%    =
    \frac{\gamma}{1+\gamma}.
\]
It is well know~(see, e.g., \cite{levin2009markov}) that the probability of this walk to hit $\kappa$ before $0$ is
\[
    \frac{1-\left(\frac{1-p}{p}\right)^{k_0}}{1-\left(\frac{1-p}{p}\right)^{\kappa}}
    =
    \frac{1-\gamma^{k_0}}{1-\gamma^{\kappa}}
    \geq 1-\gamma,
\]
where the last inequality holds because $\gamma < 1$ and $k_0\geq 1$.

Next we compute the expected time before $\CC$ hits set $\{0,\kappa\}$; let $T_k$ denote this expected time if $C_0 = k$.
For $0\leq k <\kappa$,  let $\tau_{k}$ be the expected time before $\CC$ hits $\{0,k+1\}$ if $C_0 = k$.
Then $T_k = \sum_{k\leq j<\kappa}\tau_j$.
We show by induction that
\begin{equation}
    \label{eq:tauk}
    \tau_k
    =
    \frac{1}{b(1+\gamma)}\cdot
    \sum_{1\leq i\leq k} \frac1i{\left(\frac{\gamma}{1+\gamma}\right)^{k-i}}.
\end{equation}
This is true for $k=0$, as $\tau_0=0$.
Assume it is true for some $0\leq k<\kappa-1$.
For $\tau_{k+1}$ we have
\[
    \tau_{k+1} = (1-b_{k+1}-d_{k+1})\tau_{k+1} + d_{k+1}\tau_k + 1.
\]
Solving for $\tau_{k+1}$ and substituting $b_{k+1} = bk$ and $d_{k+1} = \gamma b k$,  yields
\[
    \tau_{k+1} = \frac{\gamma}{1+\gamma} \tau_k + \frac{1}{b(1+\gamma)(k+1)}.
\]
Substituting then the value of $\tau_k$ from the induction hypothesis gives the desired expression for $\tau_{k+1}$, completing the proof of~\eqref{eq:tauk}.
We now substitute $b=1/(8N)$ in~\eqref{eq:tauk} and use that $\gamma \leq 1$ to obtain
\[
    \tau_k
    \leq
    8N\cdot
    \sum_{1\leq i\leq k} \frac1{i2^{k-i}}
    =
    8N\cdot
    \frac3k
    .
\]
The expected time before $\CC$ hits set $\{0,\kappa\}$ from $k_0\geq 1$ is then
\[
    T_{k_0}
    =
    \sum_{k_0\leq k<\kappa}\tau_k
    \leq
%    \sum_{1\leq k<\kappa}\tau_k
%    \leq
    \sum_{1\leq k<\kappa} {24N}/{k}
    \leq
    24N\log(2\kappa)
    =
    16N\log(\alpha n).
\]

We turn the above expectation result into a probabilistic bound using Markov's inequality.
If $t_{k_0}$ is the time before $\CC$ hits $\{0,\kappa\}$ from $k_0$,
\[
    \Pr\left[t_{k_0} > \frac{2T_{k_0}}{1-\gamma}\right] < \frac{1-\gamma}{2}.
\]
So, with probability at least $1 - \frac{1-\gamma}2$, the number of steps before $\CC$ hits $\{0,\kappa\}$ is at most $\frac{2T_{k_0}}{1-\gamma}\leq \frac{32Nr\log(\alpha n)}{r-1-4\alpha^{-1}}$.
Combining this with the result we showed earlier, that the probability of $\CC$ hitting $\kappa$ (before $0$) is at least $1-\gamma$, and applying a union bound, we obtain that with probability at least
$1-\gamma - \frac{1-\gamma}2 = \frac{1-\gamma}2$, $\CC$ hits $\kappa$ in at most
$\frac{32Nr\log(\alpha n)}{r-1-4\alpha^{-1}}$ steps.
This completes the proof of Claim~\ref{clm:hitting}.
\qed

\subsection{Proof of Lemma~\ref{lem:key}: Fixation and Harmonic Volume of Mutants}
\label{sec:lem:key}

%    Let $G = (V,E)$ be any connected graph with $|V|= n$ vertices
%    and minimum degree $\delta$,
%    and let $S\subseteq V$ be any set of the vertices.
%    If all vertices $u\in S$ have the mutation initially, then fixation is reached with probability at least
%    $\frac{1-r^{-\delta\cdot hvol(S)}}{1-r^{-\delta\cdot hvol(V)}}$.
We reduce the problem of finding a lower bound on the fixation probability, to that of bounding the ruin probability for a variation of the gambler's ruin problem.

We consider the following variation of the gambler's ruin problem.
The player (gambler) starts with an initial fortune of $X_0$ coins, and the initial capital of the house is $m - X_0$ coins.
The player gambles until he owns all the $m$ coins, or runs out of coins---he is \emph{ruined}.
Suppose  the fortune of the player after $i$ gambles is
$X_i = k \in\{1,\ldots, m-1\}$.
In the $(i+1)$-th gamble, the player chooses a \emph{bet} amount $b_i$ and a \emph{profit} amount $p_i$, such that
\begin{equation}
    \label{eq:bp}
    b_i,p_i\in\{1,\ldots,\sigma\},
    \qquad
    b_i \leq k,
    \qquad
    p_i \leq m-k,
\end{equation}
where $\sigma$ is a fixed \emph{maximum stake}.
If the player wins the gamble, his fortune increases to $X_{i+1} = k+p_i$, while if he loses his fortune drops to $X_{i+1} = k-b_i$.
The probability of winning is $\frac{b_i}{b_i+r p_i}$, and the probability of losing is $\frac{r p_i}{b_i+r p_i}$, where $r>1$ is fixed and determines the house edge \big(which is $\frac{r-1}{r+1}$\big).

%There is a total number of $m$ coins, some of which belong initially to the player (gambler), while the rest belong to the house.
%The player places gambles until he owns all the $m$ coins, or runs out of coins (is \emph{ruined}).
%Suppose that the fortune of the player after the first $i$ gambles is  $X_i = k\in\{1,\ldots,m-1\}$.
%In the next gamble, the player chooses a \emph{bet} amount $b$ and a \emph{profit} amount $p$, such that $b,p\in\{1,\ldots,\sigma\}$, where $\sigma$ is a fixed \emph{maximum stake}, and also $b\leq k$ and $p\leq m-k$.
%If the player wins the gamble his fortune increases to $X_{i+1} = k+p$, and if he loses it drops to $X_{i+1} = k-b$.
%The player wins the gamble with probability $\frac{b}{b+r p}$, and loses with probability $\frac{r p}{b+r p}$, where $r>1$ is the house edge.

Let $\pi_{\AA,k}$ denote the probability that the player is ruined, when he uses strategy $\AA$ to determine the bet and profit amounts in each gamble, and his initial fortune is $X_0 = k$.
We will show that for any strategy $\AA$,
\begin{equation}
    \label{eq:ruin}
    \pi_{\AA,k} \geq \frac{1-r^{-(m-k)/\sigma}}{1-r^{-m/\sigma}}.
\end{equation}
Before we prove that, we argue that it implies the lower bound of Lemma~\ref{lem:key} on the fixation probability.

Let $M_i^\ast$, for $i\geq 0$, denote the set of mutants in $G$ after $i$ steps of the Moran process, \emph{when we ignore any steps that do not change the set of mutant vertices}; the initial set of mutant vertices is $M^\ast_0 = S$.
Let $Y_i = hvol(V\setminus M_i^\ast)\cdot n!$, be the harmonic volume of the set of non-mutants after $i$ steps of this process, scaled by factor of $n!$; this factor ensures that $Y_i$ (and all quantities we specify later) take integer values.
We describe an instance of the gambler's game in terms of the above Moran process, such that the player's fortune after $i$ gambles is $X_i = Y_i$, for all $i\geq 0$.
In this game:
(i)~the initial fortune of the player is $X_0 = hvol(V\setminus S)\cdot n!=Y_0$;
(ii)~the total number of coins is $m = hvol(V)\cdot n!$;
(iii)~the maximum stake is $\sigma = n!/\delta$;
and
(iv)~for each $i\geq0$, the parameters and outcome of the $(i+1)$-th gamble are as follows.
Suppose that $X_i=Y_i$ (this is true for $i=0$; and assuming it is true for $i$, we will see it is true for $i+1$).
Consider the step from $M_i^\ast$ to $M_{i+1}^\ast$, and suppose that $u,v$ are the pair of adjacent vertices with $u\in M_i^\ast$ and $v\in V\setminus M_i^\ast$ which interact in that step, i.e., either $u$ reproduces on $v$ (and $M_{i+1}^\ast=M_i^\ast\cup\{v\}$), or $v$ reproduces on $u$ (and $M_{i+1}^\ast=M_i^\ast\setminus\{u\}$).
We set the bet and profit amounts in the $(i+1)$-th gamble to be
$b_{i} = n!/\deg_v$ and $p_i = n!/\deg_u$
(these values satisfy constraints~\eqref{eq:bp}).
The player wins the gamble if $v$ reproduces on $u$ in the step of the Moran process, and loses if $u$ reproduces on $v$.
The distribution of the outcome of the gamble is the correct one:
Given the pair $u,v$ of vertices that interact in the step of Moran process (which determine $b_i$ and $p_i$), the probability that $v$ reproduces on $u$ (and the player wins the gamble) is
$
    \frac{1/\deg_v}{1/\deg_v + r/\deg_u} = \frac{b_i}{b_i+r p_i},
$
while the probability that $u$ reproduces on $v$ (and the player loses) is
$
    \frac{r/\deg_u}{1/\deg_v + r/\deg_u} = \frac{r p_i}{b_i+r p_i}.
$
It also easy to verify that $X_{i+1} = Y_{i+1}$, for either outcome.

It follows that the fixation probability $\rho_G^S(r)$, which equals the probability that $\exists\,i, Y_i = 0$, is equal to the probability that $\exists\,i, X_i = 0$ (i.e., the ruin probability) for the game described above.
Applying now~\eqref{eq:ruin}, for $k=Y_0^\ast=hvol(V\setminus S)\cdot n!$, $m = hvol(V)\cdot n!$, and $\sigma = n!/\delta$, yields the desired lower bound on $\rho_G^S(r)$

It remains to prove~\eqref{eq:ruin}.
We observe that it suffices to consider strategies $\AA$ that are \emph{deterministic} and \emph{Markovian}, where the latter means that the choice of the bet and profit amounts in each gamble depends only on the current fortune of the player.
One can derandomize a strategy without increasing the ruin probability, by choosing at each gamble the pair of bet and profit amounts (among the $\sigma^2$ possible combinations) that minimizes the ruin probability conditionally on the past history.
It is sufficient to consider Markovian strategies because the state of the game is completely determined by the current fortune of the player.

%Note that in a general strategy, the choice of $b_i$ and $p_i$ can be a (probabilistic) function of all previous $i$ gambles, i.e., of $b_t,p_t,X_t$, for all $t<i$.
%We can build an optimal strategy $D_{opt}$ that is deterministic and Markovian, from any optimal strategy $R_{opt}$ as follows.
%$R_{opt}$ can be viewed as a probability distribution on a set of strategies in which the choice of the bet $b^{(1)}$ and profit $p^{(1)}$ amounts in the \emph{first} gamble is deterministic, i.e., it depends only on $X_0$.
%For each possible value $k\in\{1,\ldots,m\}$ for the initial fortune $X_0$ of the player, there are (at most) $\sigma^2$ possible combinations for the values of $b^{(1)},p^{(1)}$.
%Let $R_{opt,k,b,p}$ denote the strategy obtained from $R_{opt}$ for $X_0=k$, setting $b^{(1)}=b$ and $p^{(1)}=p$.
%Let $R_{opt,k}$ be one among those strategies that minimizes the ruin probability, and let $p_k,b_k$ denote the (deterministic) choices of this strategy for the first gamble.
%Strategy $R_{opt,k}$ is optimal for $X_0=k$, as $R_{opt}$ is optimal.
%Strategy $D_{opt}$ is then defined as follows:
%For each $i$ and $k$, if $X_i = k$, then $b^{(i)} = b_k$ and $p^{(i)} = p_k$.
%A simple inductive argument shows that $D_{opt}$ is an optimal betting strategy.

Consider now an arbitrary deterministic Markovian strategy $\AA$, and let $b(k), p(k)$ denote the bet and profit amounts for the strategy, when the player's fortune is $k$.
For this strategy, sequence $X_0,X_1,\ldots$ is a Markov chain with state space $\{0,\ldots,m\}$, two absorbing states: $0$ and $m$, and two possible transitions from each state $k\notin\{0,m\}$: to $k+p(k)$ with probability $\frac{b(k)}{b(k)+r p(k)}$, and to $k-b(k)$ with probability $\frac{r p(k)}{b(k)+r p(k)}$.
%Let $\pi_k$ be the ruin probability when $X_0=k$.
It follows that  $\pi_{\AA,0} = 1$, $\pi_{\AA,m}=0$, and for $0<k<m$,
\begin{equation}
    \label{eq:pi}
    \pi_{\AA,k}
    =
    \pi_{\AA,k+p(k)}\cdot\frac{b(k)}{b(k) + r p(k)}
    +
    \pi_{\AA,k-b(k)}\cdot\frac{r p(k)}{b(k)+r p(k)}.
\end{equation}
Let $\pi_k = \frac{1-r^{-(m-k)/\sigma}}{1-r^{-m/\sigma}}$.
We have $\pi_0 = 1$, $\pi_m = 0$,
and we now show that for any $0 < k < m$,
\begin{equation}
    \label{eq:pitilda}
    \pi_k
    \leq
    \pi_{k+p(k)}\cdot\frac{b(k)}{b(k)+r p(k)}
    +
    \pi_{k-b(k)}\cdot\frac{r p(k)}{b(k)+r p(k)}.
\end{equation}
In the following we will write $p$ and $b$ instead of $p(k)$ and $b(k)$.
The right side of the above inequality equals
\begin{align*}
%    &
    \frac{1-r^{-(m-k-p)/\sigma}}{1-r^{-m/\sigma}}\cdot\frac{b}{b+r p}
    +
    \frac{1-r^{-(m-k+b)/\sigma}}{1-r^{-m/\sigma}}\cdot\frac{r p}{b+r p}
%    \\
%    =\ &
%    \frac{(b+\rho p) - b\rho^{-(m-k-p)/\sigma}-\rho p \rho^{-(m-k+b)/\sigma}}{(1-\rho^{-m/\sigma})(b+\rho p)}
%    \\
%    =\ &
%    \frac{1 - \frac{b\rho^{-(m-k-p)/\sigma}+\rho p \rho^{-(m-k+b)/\sigma}}{b+\rho p}} {1-\rho^{-m/\sigma}}
%    \\
%    =\ &
    =
    \frac{1 - r^{-(m-k)/\sigma}\cdot \frac{br^{p/\sigma}+ p r^{1-b/\sigma}}{b+r p}} {1-r^{-m/\sigma}}.
\end{align*}
Comparing the quantity on the right with
$\pi_k = \frac{1-r^{-(m-k)/\sigma}}{1-r^{-m/\sigma}}$, we see it suffices to show $\frac{br^{p/\sigma}+ p r^{1-b/\sigma}}{b+r p}\leq 1$ in order to prove~\eqref{eq:pitilda}.
We do so by showing that the following real function $f(x,y)$ is non-negative:
\begin{equation}
    \label{eq:f}
    f(x,y)
    :=
    x+r y - xr^{y}- y r^{1-x}
    \geq 0,
    \qquad \text{for } 0\leq x,y\leq1.
\end{equation}
This implies $f(b/\sigma,p/\sigma) \geq 0$, for $0 < b,p \leq \sigma$,
which yields the desired result, that $\frac{br^{p/\sigma}+ p r^{1-b/\sigma}}{b+r p}\leq 1$.
We now prove~\eqref{eq:f}.
For any $x,y$, simple calculations yield $f(x,0)=f(0,y)=f(1,1)=0$.
Next we compute the partial second derivatives of $f$:
%\[
%    \tfrac{\partial}{\partial x}f(x,y)
%    =
%    \tfrac{\partial}{\partial x}\left(x+r y - xr^{y}- y r^{1-x}\right)
%    =
%    1 - r^y + yr^{1-x}\lnr,
%\]
%\[
%    \tfrac{\partial^2}{\partial x^2}f(x,y)
%    =
%    \tfrac{\partial}{\partial x}\left(1 - r^y + yr^{1-x}\lnr\right)
%    =
%    - yr^{1-x}\ln^2 r
%    \leq 0,
%\]
%\[
%    \tfrac{\partial}{\partial y}f(x,y)
%    =
%    \tfrac{\partial}{\partial y}\left(x+r y - xr^{y}- y r^{1-x}\right)
%    =
%    r - xr^y\lnr - r^{1-x},
%\]
%\[
%    \tfrac{\partial^2}{\partial y^2}f(x,y)
%    =
%    \tfrac{\partial}{\partial y}\left(r - xr^y\lnr - r^{1-x}\right)
%    =
%    - xr^{y}\ln^2 r
%    \leq 0.
%\]
\[
    \tfrac{\partial^2}{\partial x^2}f(x,y)
    =
    - yr^{1-x}\ln^2 r,
    \qquad
    \tfrac{\partial^2}{\partial y^2}f(x,y)
    =
    - xr^{y}\ln^2 r.
\]
Observe that $\tfrac{\partial^2}{\partial x^2}f(x,y)<0$ for $y>0$, and thus for any fixed $y>0$, $f(x,y)$ is a concave function of $x$.
Therefore, $f(x,1)$ is a concave function, and since $f(0,1)=f(1,1)=0$ at the boundary, we conclude that $f(x,1)\geq 0$, for $0\leq x\leq 1$.
Similarly, we have $\tfrac{\partial^2}{\partial y^2}f(x,y)<0$, for $x>0$, which implies that for any fixed $x>0$, $f(x,y)$ is a concave function of $y$.
From this, and the fact that the boundary values are $f(x,0) = 0$, and $f(x,1)\geq0$ as shown above, it follows that $f(x,y)\geq 0$, for $0\leq x,y\leq1$.
This completes the proof of~\eqref{eq:f}, and also of~\eqref{eq:pitilda}.

Next we compare~\eqref{eq:pi} and~\eqref{eq:pitilda} to show $\pi_{\AA,k}\geq \pi_k$.
Let $\Delta_k = \pi_{\AA,k} - \pi_k$.
By subtracting~\eqref{eq:pitilda} from~\eqref{eq:pi},
\begin{equation}
    \label{eq:Delta}
    \Delta_k
    \geq
    \Delta_{k+p(k)}\cdot\frac{b(k)}{b(k)+r p(k)}
    +
    \Delta_{k-b(k)}\cdot\frac{r p(k)}{b(k)+r p(k)}.
\end{equation}
We must show that $\Delta_k \geq 0$, for all $0\leq k\leq m$.
Suppose, for contradiction, that $\Delta_k < 0$ for some $k$, and let $\ell$ be such that $\Delta_\ell$ is minimal.
Note that $\ell\notin\{0,m\}$, because $\Delta_0 =1-1=0$ and $\Delta_m = 0-0=0$.
Let $\ell_0<\ell_1<\cdots<\ell_\lambda$, where  $\ell_0=\ell$, $\ell_\lambda = m$, and for each $0\leq i<\lambda$, $\ell_{i+1} = \ell_i+p_{\ell_i}$.
Applying~\eqref{eq:Delta} for $k=\ell_0$ and using the minimality of $\Delta_{\ell_0}$, we obtain that $\Delta_{\ell_1}$ is also minimal (and so is $\Delta_{\ell_0-b_{\ell_0}}$).
Repeating this argument for all $\ell_i$, yields $\Delta_{\ell_0}=\Delta_{\ell_1}=\cdots=\Delta_{\ell_\lambda}$, and thus $\Delta_m =\Delta_\ell <0$, which is a contradiction.
This proves that $\pi_{\AA,k}\geq \pi_k$ for all $0\leq k\leq m$, which
%And since we have assumed an arbitrary betting strategy, the results holds also for an optimal strategy.
proves~\eqref{eq:ruin}, and completes the proof of Lemma~\ref{lem:key}.
\qed

\section{Upper Bound for Fixation on General Graphs}
\label{sec:upperbound}

In the section we prove Theorem~\ref{thm:upperBound}, which states that for any undirected graph $G=(V,E)$ with $|V|=n$ vertices, the fixation probability is
$\rho_G(r) = 1 -  \Omega(r^{-5/3} n^{-1/3} \log^{-4/3}n)$.

We will use the following shorthand notation.
For any vertex set $S\subseteq V$, we will write $\rho_S$ and $\zeta_S$ instead of $\rho_G^S(r)$ and $\zeta_G^S(r)$, respectively, for the probabilities of fixation and extinction in $G$, when the initial set of mutants is $S$.
Further, when $S=\{u\}$ is a singleton set, we will write just $\rho_u$ and $\zeta_u$.
By $\rho$ and $\zeta$ we will denote the fixation and extinction probabilities for a uniformly random initial mutant.

%Let $\rho_u$ denote the probability that fixation is reached when the initial mutant is vertex $u$, and
%let $\zeta_u = 1-\rho_u$ be the corresponding extinction probability.
%We extend this notation to the case where a \emph{set} $S\subseteq V$ of vertices is  initially infected (rather than a single vertex), and write $\rho_S,\zeta_S$ to denote the fixation and extinction probabilities in this case.
%As before, $\rho$ and $\zeta$ are the fixation and extinction probabilities for a uniformly random initial mutant.

We start with a basic general lower bound on $\zeta_u$, in terms of the distance-2 neighborhood of $u$ (precisely, in term of the degrees of $u$'s neighbors and their neighbors).
For any pair of vertices $u\in V$ and $v\in \Gamma_u$,
let $T_u$ be the harmonic volume of $u$'s neighbors, and $T_{u,v}$ the harmonic volume of $v$'s neighbors other than $u$:
\[
    T_u = hvol(\Gamma_u),
    \qquad
    T_{u,v}
    =
    hvol\left(\Gamma_v\setminus\{u\}\right)
    =
    T_v - (\deg_u)^{-1}.
\]

\begin{claim}
    \label{clm:zetalb}
    For any vertex $u\in V$, we have
    $\displaystyle
        \zeta_u
        \geq
        \frac{T_u}{T_u+\frac r{\deg_u} \sum_{v\in \Gamma_u}\frac{2r}{2r+T_{u,v}}}.
    $
\end{claim}
\begin{proof}
%We lower bound $\zeta_u$ by the probability that all mutant offsprings of $u$ die before they reproduce.
Let $\deg_u = d$, and let $v_1,\ldots,v_d$ denote $u$'s neighbors.
%Let $\zeta_{u,2} = \frac1d\sum_{1\leq i\leq d}\zeta_{\{u,v_i\}}$, be the probability of extinction when $u$ and a uniformly random neighbor $v_i$ of $u$ are mutants.
We have
\[
    \zeta_u = \frac{T_u}{r+T_u }  +  \frac {r/d}{r+T_u }\sum_{1\leq i\leq d}\zeta_{\{u,v_i\}},
\]
where $\frac{T_u}{r+T_u }$ is the probability that $u$ dies before it reproduces, and $\frac {r/d}{r+T_u }$ is the probability it reproduces on any given neighbor.
We lower bound $\zeta_{\{u,v_i\}}$ by assuming it is $0$ unless $v_i$ is replaced by a non-mutant before any of $u,v_i$ reproduce:
\[
    \zeta_{\{u,v_i\}}
    \geq
    \frac{T_{u,v_i}}{2r + T_{u,v_i}}\cdot \zeta_u.
\]
Substituting this in the expression for $\zeta_u$ above and solving for $\zeta_u$ completes the proof.
\end{proof}

Note that  Claim~\ref{clm:zetalb} implies the weaker bound of $\zeta_u \geq \frac{T_u}{T_u+r}$, which considers only the immediate neighbors of~$u$.

Next we bound the number of vertices of degree at least $d$ in terms of $\zeta$.
Let $n_d = |\{u\in V\colon \deg_u = d\}|$
be the number of vertices with degree exactly $d$,
and let
$N_d = \sum_{d'\geq d} n_{d'}$.
%be the number of vertices of degree at least $d$.
Let
\[
     \theta = \max\{d \colon N_{d} \geq d/2\}.
\]
Note that $N_{\theta+1} < (\theta+1)/2$, and thus $N_{\theta+1} \leq \theta/2$.
This implies that at least half of the neighbors of any vertex with degree at least $\theta$ have degree at most $\theta$.
%Let
%\[
%    \theta = \lfloor(4r+2)n\zeta\rfloor,
%    \qquad
%    \eta = (32r^3 +24r^2+8r+2) n^2\zeta^2.
%\]
%We show that at most $\theta/2$ vertices have degree larger than $\theta$, and that the number of vertices of degree at least $k$ is at most $\eta/k$, for $k\leq \theta$.
%Note that the latter bound is useful only when $k\geq \eta/n$ (otherwise, $\eta/k\geq n$.)

\begin{claim}
    \label{clm:N}
%    {\bf(a)}~$\theta \leq (4r+2) n\zeta$;
%    {\bf(b)}~for $d \leq \theta$, $N_{d} \leq \frac{(4r^2+1)\theta n\zeta}d$; and
%    {\bf(c)}~for $d > \theta$, $N_d\leq \frac{(8r^2+2)\theta n\zeta\cdot \ln(\theta+1)}d$.
    {\bf(a)}~$\theta \leq 6r n\zeta$; \
    {\bf(b)}~for $d \leq \theta$, $N_{d} \leq \frac{4r^2 \theta n\zeta}d$; and \
    {\bf(c)}~for $d > \theta$, $N_d\leq \frac{8r^2\theta n\zeta\cdot \log(2\theta)}d$.
\end{claim}

Observe that, if we ignore logarithmic factors and assume $r$ is bounded, the $\tilde\Omega(n^{-1/3})$ lower bound we wish to show for~$\zeta$, together with Claim~\ref{clm:N}, suggests $\theta = \tilde O(n^{2/3})$, and $N_{d} = \tilde O(n^{4/3}/d)$, for any $d$.

\begin{proof}[Proof of Claim~\ref{clm:N}]
{(a):}
For any vertex $u\in V$ of degree $d\geq \theta$,
\begin{equation}
    \label{eq:TuLarge}
    T_u
    \geq
    \sum_{v\in \Gamma_u: \deg_v\leq \theta} \frac1{\deg_v}
    \geq
    \frac{d-N_{\theta+1}}{\theta}
    \geq
    \frac {d-\theta/2}{\theta}
    \geq
    1/2,
\end{equation}
where for the second-last inequality  we used that $N_{\theta+1} < (\theta+1)/2$ and thus $N_{\theta+1}\leq \theta /2$.
Combining the above with $\zeta_u \geq \frac{T_u}{r+T_u }$, which follows from Claim~\ref{clm:zetalb}, yields
$
    \zeta_u
%    \geq
%        \frac{T_u}{T_u+r}
    \geq
    \frac{1}{2r+1}.
$
From this, we get
\[
    \sum_{u:\deg_u\geq \theta} \zeta_u
    \geq
    \frac{N_{\theta}} {2r + 1},
\]
and since the sum on left can be at most $n\zeta$, we have
$
    N_{\theta}
    \leq
    (2r+1)n\zeta.
$
Combining this with $N_{\theta}\geq \theta/2$, yields $\theta \leq 2(2r+1)n\zeta\leq 6rn\zeta$.

{(b):}
Let $u\in V$ and $k=\deg_u$.
Let $v_1,\ldots,v_k$ be the vertices adjacent to $u$, and let $d_1,\ldots, d_k$ be their respective degrees.
We argue that for each vertex $v_i$,
\[
    T_{u,v_i}
    \geq
    \frac{2d_i}{3\theta} - 1:
\]
This is trivially true if $d_i\leq 3\theta/2$, because the right side is at most $0$;
if $d_i>3\theta/2$, the number of neighbors of $v_i$ with degree at most $\theta$ is at least $d_i - N_{\theta+1} \geq  d_i - \theta/2 \geq 2d_i/3$, thus $T_{u,v_i} \geq \frac{2d_i/3 - 1}\theta \geq \frac{2d_i}{3\theta} - 1$.
Applying now Claim~\ref{clm:zetalb}, and using the inequality we have just showed, yields
\[
    \label{eq:zetau}
    \zeta_u
    \geq
    \frac{T_u}{T_u+\frac r{k} \sum_{1\leq i\leq k}\frac{2r}{2r+T_{u,v_i}}}
    \geq
    \frac{T_u}{T_u+\frac r{k} \sum_{1\leq i\leq k}\frac{2r}{2d_i/(3\theta)}}
    =
    \frac{T_u}{T_u+\frac r{k} (3r\theta\cdot T_u)}
    =
    \frac{k}{k+3r^2\theta}
    .
\]
The quantity on the right is at least $\frac k {(3r^2+1)\theta}$ if $k\leq \theta$,
and at least $\frac 1{(3r^2+1)}$ if $k\geq \theta$.
It follows that for $d\leq \theta$,
\[
    \sum_{u:\deg_u\geq d} \zeta_u
    \geq
    N_d\cdot \frac{d}{(3r^2+1)\theta},
\]
and since the sum on the left is at most $n\zeta$, we conclude that
$
    N_{d}
    \leq \frac{(3r^2+1)\theta n\zeta}d
    \leq \frac{4r^2 \theta n\zeta}d.
$

(c):
Let $D = vol(\{u\colon \deg_u\leq \theta\})$
be the sum of the degrees of all vertices that have degree at most $\theta$.
We first bound~$D$ using the result from part~(b):
\[
    D=
    \sum_{1\leq d\leq\theta} (N_d-N_{d+1})\cdot d
    \leq
    \sum_{1 \leq d\leq\theta} N_d
    \leq
    \sum_{1 \leq d\leq\theta} \frac{4r^2 \theta n\zeta}d
    \leq
    4r^2\theta n\zeta\cdot \log(2\theta).
\]
Next we observe that for $d>\theta$, any vertex $u\in V$ with degree $\deg_u\geq d$ has at least $d-N_{\theta+1} \geq d-\theta/2 > d/2$ neighbors of degree at most $\theta$.
It follows that $D\geq N_d\cdot d/2$.
Combining this with the upper bound for $D$ we computed above, yields
$N_d\leq \frac{8r^2\theta n\zeta\cdot \log(2\theta)}d$.
\end{proof}

An immediate corollary of Claim~\ref{clm:N}(a) and~(b), stated next, is that
$\rho \leq 1 - \Omega\left(r^{-3/2}n^{-1/2}\right)$.
This is a weaker bound than the one of Theorem~\ref{thm:upperBound}, but is already an improvement over the previous best bound~\cite{mertzios2013strong}.

\begin{corollary}
    For any undirected graph $G$ on $n$ vertices,
    %and minimum degree $\delta$,
    the fixation probability is at most $1 - \frac{1}{5\sqrt{r^3n}}$.
\end{corollary}
\begin{proof}
From Claim~\ref{clm:N}(b) it follows that $N_1\leq 4r^2\theta n\zeta$.
Substituting $\theta \leq 6r n\zeta$, by Claim~\ref{clm:N}(a),  and $N_1=n$, we obtain $\zeta \geq 1/\sqrt{24r^3 n}$.
\end{proof}

We are now ready to prove our main result (this proof is a bit technical).

\subsection{Proof of Theorem~\ref{thm:upperBound}}
\label{sec:proofUB}

Let $V_d = \{u\in V\colon d \leq \deg_u < 2d \}$, and consider an $\ell\in\{1,\ldots,n\}$ for which
\begin{equation}
    \label{eq:Vell}
    |V_\ell|\geq \frac n{1+\log n} = \frac n{\log (2n)}.
\end{equation}
Such an $\ell$ clearly exists, because $|V|=n$ and $V = \bigcup_{0\leq i \leq \log n} V_{2^i}$.
For $0\leq i \leq \log n$, let
\[
    H_i = \left\{v\in V_{2^i} \colon T_v\geq
    \frac{n}{4\ell |V_{2^i}|\log^2 (2n)}\right\},
\]
and let  $H = \bigcup_{i} H_{i}$.
Finally,
\[
    L =\left\{u\in V_\ell \colon \zeta_u \leq 1/2,\, \Gamma_u\subseteq H \right\}.
\]

Next we show that $L$ is not much smaller than $V_\ell$.
For that, we first argue that $\sum_{v\in H} T_v$ is very close to $\sum_{v\in V} T_v = n$.
We have
\begin{align*}
    \sum_{v\notin H} T_v
    =
    \sum_{i}\sum_{v\in V_{2^i}\setminus H_{i}} T_v
    &
    \leq
    \sum_{i}\left(
        |V_{2^i}\setminus H_{i}|\cdot\frac{ n}{4\ell |V_{2^i}|\log^2 (2n)}
    \right)
    \\&
    \leq
    \sum_{0\leq i\leq \log n}
    \frac{n}{4\ell\log^2 (2n)}
    \leq
    \frac{n}{4\ell\log(2n)}.
\end{align*}
To lower bound $|L|$, we observe that the number $k$ of vertices $u\in V_\ell$ with $\Gamma_u\not\subseteq H$, and thus with at least one neighbor in $V\setminus H$, must satisfy  $\frac{k}{2\ell} \leq \sum_{v\notin H}$;
and using the bound on $\sum_{v\notin H}$ we showed above, we get
$k\leq \frac{n}{2\log(2n)}$.
Also for the number $k'$ of vertices $v$ with $\zeta_v> 1/2$, we have $k'/2 \leq n\zeta$, and thus $k' \leq 2n\zeta$.
It follows
\begin{equation*}
%    \label{eq:cardL}
    |L|
    \geq
    |V_\ell| - \frac n{2\log(2n)} - 2n\zeta
    \geq
    \frac n{2\log(2n)} - 2n\zeta,
\end{equation*}
by~\eqref{eq:Vell}.
We can assume that $\zeta \leq \frac 1{8\log(2n)}$, because otherwise the theorem holds trivially.
Then
\begin{equation}
    \label{eq:cardL}
    |L|
    \geq
    \frac n{4\log(2n)}.
\end{equation}

Next we show a lower bound on $\zeta_u$, for $u\in L$.
From Claim~\ref{clm:zetalb}, for any vertex $u\in L$,
\begin{align*}
    \zeta_u
    \geq
    \frac{T_u}{T_u+\frac r{\deg_u} \sum_{v\in \Gamma_u}\frac{2r}{2r+T_{u,v}}}
%    \\&
    &\geq
    \frac{T_u}{2\frac {r}{\deg_u} \sum_{v\in \Gamma_u}\frac{2r}{2r+T_{u,v}}}
    \geq
    \frac{T_u}{2\frac {r}{\ell}\sum_{v\in \Gamma_u}\frac{2r}{T_{v}}}
    =
    \frac{\ell\cdot T_u}{4r^2\sum_{v\in \Gamma_u}T_{v}^{-1}},
\end{align*}
where for the second inequality we used that $\zeta_u\leq 1/2$ to obtain that $T_u$ is at most half the denominator.
Let $\lambda_i = |\Gamma_u \cap H_{i}|$ be the number of neighbors of $u$ in $H_{i}$, and recall that $\Gamma_u\subseteq H$.
We have
\begin{gather*}
    T_u \geq \sum_i \frac{\lambda_i}{2^{i+1}},
    \text{ and }
%    \\
    \sum_{v\in \Gamma_u}T_{v}^{-1}
    \leq
    \sum_i \left(\lambda_i\left(\frac{n}{4\ell |V_{2^i}|\log^2 (2n)}\right)^{-1}\right).
%    =
%    \sum_i \frac{2\lambda_i \ell |V_{2^i}|\log (2n)}{n}.
\end{gather*}
Substituting these above gives
\[
    \zeta_u
    \geq
    \frac{\sum_i ({\lambda_i}/{2^{i+1}})}
    {16r^2\log^2 (2n)\sum_i (\lambda_i |V_{2^i}|)/n}.
\]
Next we use Claim~\ref{clm:N} to bound $|V_{2^i}|$.
We have
\[
    |V_{2^i}|
    \leq
    N_{2^i}
    \leq
    \frac{8r^2\theta n\zeta\cdot \log(2\theta)}{2^i}
    \leq
    \frac{48r^3 n^2\zeta^2 \log(2n)}{2^i},
\]
where the first inequality follows from Claim~\ref{clm:N}(b) and (c), and the second from Claim~\ref{clm:N}(a).
Substituting above yields
\[
    \zeta_u
    \geq
    \frac{\sum_i (\lambda_i/2^{i+1})}
    {768r^5 n\zeta^2 \log^3 (2n) \sum_i (\lambda_i/2^i)}
    =
    \frac{1}
    {1536r^5 n\zeta^2 \log^3 (2n)}.
\]
Summing now over all $u\in L$ and using that
$\sum_{u\in L}\zeta_u \leq n\zeta$,
we obtain
\[
    \frac{|L|}{1536r^5 n\zeta^2 \log^3 (2n)}
    \leq
    n\zeta.
\]
Substituting $|L|$ by its the lower bound $\frac n{4\log(2n)}$ from~\eqref{eq:cardL} and solving for $\zeta$ yields
\[
    \zeta\geq \frac{1}{19r^{5/3} n^{1/3} \log^{4/3} (2n)}.
\]
This completes the proof of Theorem~\ref{thm:upperBound}.
\qed

\section{Conclusion}
\label{sec:conclusion}

We have presented the first examples of non-directed networks that act as strong amplifier of selection or strong suppressor for the classic Moran process on structured populations.
This process is used to model the spread of mutations in populations of organisms, as well as somatic evolution in multicellular organism.
Strong amplifiers have the surprising property that an advantageous mutation will almost surely spread to the whole network, however small its selective advantage; while in strong suppressors, the selective advantage of a mutation has negligible effect on the probability it prevails.
The Moran process, can also be viewed as a basic, general diffusion model, for instance, of the spread of ideas or influence in a social network.
In this context, our results suggest network structures which are well suited for the spread of favorable ideas, and structures where the selection of such ideas is suppressed.

We have showed that the our strong amplifier is existentially optimal.
However, it is an open problem whether better suppressors that ours exist.
Another direction of interest would be to investigate the relationship between fixation probability and absorption time.
For instance, high fixation probability seems to imply large absorption times.
Our examples of strong amplifiers and suppressors are shrinkingly simple.
This gives hope that they may indeed capture some aspects of real evolutionary networks.
Validating our models is an interesting future direction.

\subsection*{Acknowledgements}

I would like to thank Christoforos Raptopoulos and George Mertzios for bringing to my attention the question about the existence of undirected strong amplifiers and suppressors, during ICALP'14.
%and also Panayotis Tsaparas and Vasiliki Bartzoka for helpful discussions.

%\setlength{\bibsep}{0.6ex}
\bibliographystyle{plain}
\bibliography{moran-fixation}

\end{document}